\documentclass{article}

\usepackage[a4paper, margin=1in]{geometry}
\usepackage[utf8]{inputenc}
\usepackage{amsmath}
\usepackage{amsthm}
\usepackage{amsfonts}
\usepackage{comment}
\usepackage{algorithm}
\usepackage[noEnd=true]{algpseudocodex}
\usepackage{color}
\usepackage{mathtools}
\usepackage{tikz}
\usepackage{fancybox} 
\usepackage{tikz-qtree}
\usepackage{hyperref}
\usepackage{cleveref}
\usepackage{physics}
\usetikzlibrary{quantikz} \usepackage{enumitem}

\newtheorem{theorem}{Theorem}[section]
\newtheorem{proposition}[theorem]{Proposition}

\newtheorem{definition}[theorem]{Definition}

\newcommand{\RR}{\mathbb{R}}
\newcommand{\CC}{\mathbb{C}}

\newcommand{\pop}[1]{\mathsf{#1}}
\newcommand{\nop}[1]{\widetilde{\mathsf{#1}}}
\newcommand{\ops}[1]{\mathbb{#1}}

\begin{document}

\title{Low bit-flip rate probabilistic error cancellation}

\author{Mathys Rennela, Harold Ollivier\\
  QAT, DIENS, Ecole Normale Supérieure, 45 rue d'Ulm, 75005 Paris\\
  Universite PSL, CNRS, INRIA}

\maketitle 

\begin{abstract}
  Noise remains one of the most significant challenges in the development of reliable and scalable quantum processors. While quantum error correction and mitigation techniques offer potential solutions, they are often limited by the substantial overhead required. To address this, tailored approaches that exploit specific hardware characteristics have emerged. In quantum computing architectures utilizing cat-qubits, the inherent exponential suppression of bit-flip errors can significantly reduce the qubit count needed for effective error correction. In this work, we explore how the unique noise bias of cat-qubits can be harnessed to enhance error mitigation efficiency. Specifically, we demonstrate that the sampling cost associated with probabilistic error cancellation (PEC) methods can be exponentially reduced with the depth of the circuit when gates act on cat-qubits and preserve the noise bias. Similar results also hold for Clifford circuits and Pauli channels. Our error mitigation scheme is benchmarked across various quantum machine learning circuits, showcasing its practical advantages for near-term applications on cat-qubit architectures.
\end{abstract}

\section{Introduction}
\label{sec:intro}

\paragraph{Context.}
Quantum error correction techniques aim to undo computational errors that unavoidably arise during the execution of quantum circuits. However, the physical qubit count overhead they impose make them impractical for smaller machines. In contrast, quantum error mitigation techniques focus on improving the accuracy of expectation values obtained by sampling the output from quantum circuits, thereby shifting the cost of handling errors from a space overhead to a time overhead. Although the scaling of the overhead for error mitigation is not favorable for larger circuits, it is currently the best approach to address noise for small quantum devices and unlocking their computational power for near term applications.

One such technique is Probabilistic Error Cancellation (PEC), which approximates the expectation value of an observable measured on a quantum state prepared by a noiseless circuit as the average value of measurement results taken over a quasi-probabilistic distribution of noisy quantum circuits (see \cite{zne-pec} for the original techniques and, e.g.~\cite{practical-pec,optimal-pec,limits-pec,ECBY21hybrid} for improvements and guarantees). By replacing the execution of a noiseless circuit with that of noisy but easily implementable ones, PEC yields an unbiased estimator for the expectation value of observables on the noiseless circuits.

Thanks to its appealing feature of being hardware and algorithm agnostic, PEC is easy to use and has been successfully applied to various use cases, including superconducting quantum architectures with crosstalk \cite{BMKT23probabilistic}, mitigation of Markovian noise \cite{EBL18practical}, and learning the noise associated with a given system \cite{SQCB21learning}.

Yet, this also means that PEC remains unaware of the error propagation structure within the quantum circuits of interest, potentially leading to suboptimal performance. This raises an important question, posed in \cite[Open Question 3]{CBBE23quantum}: can PEC be further optimized by tailoring it to a specific hardware or algorithm? Among current hardware proposal, cat-qubits---superconducting qubits with highly biased noise, considered as promising candidates for building fault-tolerant architectures~\cite{QCwithCats}---are of great interest. Can their unique properties be leveraged to reduce the overhead of error mitigation? More broadly, do cat-qubits offer advantages outside the fault-tolerant regime and into the NISQ regime, where error mitigation is crucial for unlocking the potential of near-term quantum devices?

\paragraph{Overview of results.}
In this paper, we provide a positive answer to the previous open questions. We do so by introducing a variant of Probabilistic Error Cancellation (PEC) that outperforms the standard approach in terms of sampling cost overhead for achieving a given target variance for the expectation value of an observable measured at the output of the noiseless circuit.

Indeed, the overhead of PEC is due to the necessity to sample noisy circuits to construct an approximation of an ideal noiseless circuit given as input. More precisely, when the ideal circuit is provided as a sequence of layered gates, PEC constructs a collection of noisy circuits from each layer by independently adding extra gates that mitigate the effect of the noise on average. This results in a quantum sampling cost that increases exponentially with the number of layers.

To address this limitation, we propose a modified error mitigation procedure, which we refer to as \emph{Block-PEC}. By grouping the added gates together and commuting them through the various noisy layers of the circuit, we can combine them into a single additional layer, effectively performing PEC through block sampling. This allows us to reduce the base and exponent of the exponential quantum sampling cost, trading part of the quantum sampling cost for a classical pre-processing step that calculates quasi-proba\-bilities. As a result, Block-PEC achieves a quantum sampling cost exponential in the qubit count of the circuit, rather than exponential in both the qubit count \emph{and} the circuit depth. We complement our theoretical analysis by giving performance guarantees under various biased noise models as well as numerical evalutation of the gain for several concrete circuit classes.

By considering circuits that are bias-preserving on a given subspace of the full Hilbert space, we can further relax the requirements of Block-PEC. This allows us to demonstrate that Block-PEC can be effectively applied to circuits with biased noise. Our numerical simulations using small-size circuits show that Block-PEC provides a clear improvement for parameterized quantum circuits formed by matchgates, particularly Reconfigurable Beam Splitter (RBS) gates, which are widely used in quantum machine learning algorithms for NISQ computers~\cite{unary-option-pricing,rbs-pde,low-depth-qae,qdeep-hedging,qvt}.

\paragraph{Related work.}
PEC is a well established error mitigation technique available in various open source packages (see e.g.~\cite{mitiq}). From~\cite{optimal-pec}, we know that standard PEC is optimal (in terms of sampling cost) when applied on one $n$-qubit layer subjected to dephasing noise. Moreover, it has been established that using a continuous set of ideal gates, containing all $Z$-axis rotations for arbitrary angles, would not help lower the cost of PEC, as explained in~\cite[Sec.~5]{optimal-pec}.

Previous works have nonetheless provided improvements to PEC by optimizing the sampling cost taking into account the gates within the past light-cone of a given observable, as only those contribute to the errors affecting the measured expectation value~\cite{local-qem}.

\paragraph{Outline.}
After an overview of the mathematical formalism of probabilistic error cancellation (Section~\ref{sec:preliminaries}), we introduce Block-PEC, and we outline the conditions under which it can be implemented (Section~\ref{sec:block-pec}). We then provide performance guarantees for Block-PEC (Section~\ref{sec:guarantees}), followed by a numerical analysis of the implementation of Block-PEC on concrete applications (Section~\ref{sec:applications}). We conclude in Section~\ref{sec:conclusion}, focusing on the prospects of this approach, and placing it in the context of hardware-specific error mitigation techniques.

\section{Preliminaries}
\label{sec:preliminaries}

\subsection{Quantum channels}
\label{sub:channels}

Quantum channels are completely positive trace preserving operators acting on density matrices. They will be written using sans-serif fonts such as $\pop C$. A unitary channel associated to a unitary operator $U$ is defined as $\pop U(\cdot) = U \cdot U^\dagger$. By convenience, we will often refer to $\pop U$ as a unitary operator.

While noise channels can in principle correspond to any CPTP maps, we restrict our study to convex combinations of unitary channels $\pop N = \sum_{\pop V \in \ops V} \alpha(\pop V) \pop V$ with $\alpha(\pop V)\geq 0$, $\sum_{\pop V} 
\alpha(\pop V) = 1$, and where $\ops V$ is a set of unitary channels (e.g. Pauli channels where the unitary belongs to the Pauli group). 

\begin{definition}[Noisy unitary channel]
The noisy version of a unitary channel $\pop U$ under the action of the noise channel $\pop N$ is written 
\begin{align}
    \widetilde{\pop U} = \pop N \circ \pop U
    = \sum_{\pop V \in \ops V} \alpha(\pop V) \pop V \circ \pop U.
\end{align}
Throughout the rest of the manuscript ``tilded'' channels refer to their noisy version.
\end{definition}

A layered circuit is a circuit which can be defined as a succession of unitaries $U=U_{d} \ldots U_{1}$ so that $d$ is the depth of the circuit, where $U_{l}$ is the unitary corresponding to the $l$-th layer of the circuit. With our notation,  $\pop U_{l}$ is the channel for the $l$-th layer, so that the full circuit is a unitary channel $\pop U = \pop U_d \circ \cdots \circ \pop U_1$. The noisy variant of the quantum channel $\pop U$ is given by
\begin{align}
    \widetilde{\pop U} 
    = \sum_{\pop V \in \ops V_1 \times \cdots \times \ops V_d}\alpha(\pop V) \pop V_d \circ \pop U_d \circ \cdots \circ \pop V_1 \circ \pop U_1,
\end{align}
where $\alpha(\pop V) = \prod_{l = 1}^d \alpha_l(\pop V_l)$ for $\pop V = (\pop V_1,\ldots,\pop V_d)$,  and each $\ops V_l$ defined as the generating set of unitaries associated to the noise channel $\pop N_l = \sum_{\pop V_l \in \ops V_l} \alpha_{l}(\pop V_l) \pop V_l$.

\subsection{Cat-qubits and dephasing noise}
\label{sub:cats}

Cat-qubits are superconducting qubits which can be physically implemented in such a way that bit flip errors are exponentially suppressed, at the cost of a linear increase of the phase flip error rate (see~\cite{QCwithCats} for a detailed presentation of cat-qubits). In other words, cat-qubits are associated with a biased noise model, heavily biased towards phase flip errors. 

There is experimental evidence that phase flips errors are drastically rarer than bit flip errors ($10^8$ times rarer in the experiments in~\cite{100secs-bitflip}, for example). In the present work, we choose to ignore bit flips altogether (we further discuss this choice in Appendix~\ref{appendix:dephasing}). To make our results more concrete, we make use of two simplified dephasing noise models.

The first one inserts correlated phase flips immediately after each gate of the circuit. More precisely, we chose to consider that after a gate acting on the set of $n$ qubits $A$, the following channel acts upon these qubits:
\begin{align}
    \pop C^{[A]}_p = (1-p)\pop I^{\otimes n} + \sum_{I^{\otimes n} \neq B \subseteq A} \frac{1}{2^n-1} p Z_{[B]}, 
\end{align} 
where 
$p$ is the $1$-qubit phase flip error probability, and $Z_{[B]}$ is the Pauli unitary operator defined for any set $B \subseteq A$ of qubits by
\begin{align}
    Z_{[B]} = \left(\bigotimes_{i \in B} Z_i\right)
    \otimes \left(\bigotimes_{i \in A \setminus B} I_i\right),
\end{align}
so that $Z_\emptyset = I^{\otimes n}$. Such correlations have been observed in cat-qubit architectures for 1 and 2-qubit gates (see e.g.~\cite{QCwithCats}).

The second model is a more drastic approximation but yields simpler analytic calculations.  It consists of applying single qubit dephasing channels on all the qubits a gate is acting upon. In effect, for a gate acting on the set of qubits $A$, the channel applied after the gate is
\begin{align}
    \pop F_p^{[A]} = \sum_{B \subseteq A:|B|=m} (1-p)^{n-m}p^m \pop Z_{[B]}.
\end{align}
In practice, this latter model neglects correlations between the errors within the set of qubits affected by the gate. In Appendix~\ref{appendix:detailed-performance-analysis}, we justify that Block-PEC retains its performance when error correlations are taken into account.

\subsection{Overview of Probabilistic Error Cancellation}
\label{sub:overview-pec}

To compensate the effect of noise, error mitigation techniques such as PEC intend to apply the inverse of the noise channel to the affected quantum systems. While it is not a physically realizable operation as soon as the noise channel is not unitary, it is nonetheless well defined mathematically. 

Indeed (see~\cite{zne-pec,practical-pec} and also Appendix~\ref{app:pec}), for large classes of noise models, one can rewrite the channel $\pop U$ corresponding to a noiseless $d$-layered circuit as a quasi-probabilistic distribution of noisy circuits~\cite{practical-pec,zne-pec}:
\begin{align}
    \pop U = \sum_{\pop V \in \ops V} \alpha(\pop V)
    \pop N_d \circ \pop V_d \circ \pop U_d \circ \cdots \circ \pop N_1 \circ \pop V_1 \circ \pop U_1,
\end{align}
where $\alpha(\pop V)= \alpha_{1}(\pop V_1) \cdots \alpha_{d}(\pop V_d)$ for $\pop V = (\pop V_1,\ldots,\pop V_d) \in \ops V$ where the $\alpha_{l}(\pop V_l)$'s are
given by the resolution of the system of equations defined for each layer by
\begin{align}
    \pop I^{\otimes n_l} = \sum_{\pop V_l \in \ops V_l} \alpha_{l}(\pop V_l)
    \pop N_l \circ \pop V_l,
\end{align}
where $\sum_{\pop V_l} \alpha_{l}(\pop V_l) = 1$, and $n_l$ is the number of qubits that the $l$-th layer $\pop U_l$ acts non-trivially on. 

Under the assumption that the noisy unitary channels $\widetilde{\pop B_l} = \pop N_l \circ \pop V_l \circ \pop U_{l}$ correspond to physically implementable gates, 
PEC determines the expectation value of the noiseless circuit $\pop U$ by randomly sampling a noisy gate $\widetilde{\pop B_l}$ for each layer, and executing the circuit corresponding to $\widetilde{\pop B_d} \cdots \widetilde{\pop B_1}$.

The noiseless expectation value $\langle O\rangle_{\text{ideal}} = \Tr[O\pop U(\rho_0)]$ for some observable $O$ (without loss of generality, $\|O\|_\infty = 1$), and given the initial state $\rho_0 = \ketbra{0}{0}^{\otimes n}$ is such that
\begin{align}
    \langle O\rangle_{\text{ideal}} = \sum_{\pop V_1,\ldots,\pop V_d \in \ops V} 
    \alpha_{\pop V_1}\cdots\alpha_{\pop V_d} \cdots
    \langle O\rangle_{\pop V_1,\ldots,\pop V_d}.
\end{align}
where
\begin{align}
    \langle O\rangle_{\pop V_1,\ldots,\pop V_d} = \Tr[O \widetilde{\pop B_{d}} \cdots \widetilde{\pop B_{1}}(\rho_0)].
\end{align}

Hoeffding's inequality states that 
the number of samples $S$ required to attain precision $\delta$ with PEC, with probability at least $1-\varepsilon$, is $S \geq \frac{\gamma_{\text{total}}^2}{2\delta^2} \ln(\frac{2}{\varepsilon})$, where $\gamma_{\text{total}} = \gamma_d\cdots\gamma_1$.
In this context, the sampling cost $\gamma_{\text{total}}$ entirely determines the performance of PEC as an error mitigation technique.

In a conventional use of PEC, the sampling cost $\gamma_{\text{total}}$ is in $O(\gamma^{nd})$, where $n$ is the qubit count, $d$ is the circuit depth, and $\gamma \simeq 1+kp+O(p^2)$ for some constant $k$ determined by the noise model (for dephasing noise, $k=2$), with $1$-qubit error rate $p$. This makes PEC a reasonable mitigation strategy in the presence of weak gate noise~\cite{optimal-pec}. 

\subsection{Inverting Pauli noise}
\label{sub:inverting-pauli-noise}

As outlined in the previous section, the key requirement for PEC is the ability to write perfect gates as a quasi-probabilistic combination of noisy ones. Whenever the noise channel applied after each ideal gates is a convex combination of Pauli channels, this can be performed rather straightforwardly by writing the inverse as a Taylor series.

Additionally, and this will be crucial for Block-PEC, whenever the above convex combination can be restricted to subgroup of the Pauli group, so does the quasi-probabilistic distribution of the inverse.

\begin{theorem}[Adapted from~\cite{ECBY21hybrid}]\label{thm:invert}
  The inverse of a Pauli-noise channel $\pop N = (1-p)\pop I + p\pop A$, where $\pop A$ is a convex combination over Pauli channels, is a quasi-probabilistic distribution of Pauli channels.

  Furthermore, if the non trivial Pauli-channels in $\pop A$ correspond to Pauli operators  that form a group, so do the non trivial terms in the quasi-probabilistic distribution of $\pop N^{-1}$.

  $\pop N^{-1}$ can be approximated by a channel of the form $(1+p)\pop I - p\pop A$.
\end{theorem}

\begin{proof}
  Consider a noise channel $\pop N = (1-p)\pop I + p \pop A = (1-p)\pop I + \sum_{\pop Q \in \mathbb Q} p_{\pop Q} \pop Q$, where each $\pop Q$ is a Pauli channel associated to a probability $p_{\pop Q}$, so that $\sum_{\pop Q} p_{\pop Q} = 1$, and such that $ \ops Q$ is a subgroup of the Pauli group.

  Then,
  \begin{align}
    \pop N^{-1}
    & = \left((1-p)\pop I + p\pop A \right)^{-1} \\
    & = (1-p)\left( \pop I +  \sum_{i>0}\left(  -\frac{p}{1-p} \pop A \right)^{i} \right) \\
    & \approx (1+p) \pop I -p\pop A.\label{eq:taylor}
  \end{align}

  The first and second claim of the theorem follow from the group structure of $\ops Q$  which guarantee that each term in the expansion of $\pop A^{i} = \left( \sum_{\pop Q \in \mathbb Q} p_{\pop Q} \pop Q\right)^{i}$ is in $\ops Q$.

  The third claim (Eq.~\ref{eq:taylor}) is obtained by truncating the Taylor series to order 1.
\end{proof}

Note that this inversion assumes that the noise channel $\pop N$ is close to the identity. In accordance with this assumption, the gates $\widetilde{\pop B_l}$ that will be used to invert $\pop N$ will introduce a noise that will be of similar magnitude. This means, following Eq,~\ref{eq:taylor}, that this effect will be of order 2 and can thus be neglected in our computation.

\section{Block Probabilistic Error Cancellation}
\label{sec:block-pec}

The standard PEC procedure outlined in Section~\ref{sub:overview-pec} is oblivious to the structure of the implementable gate set and noise model save for the actual quasi-probabilistic decomposition of each noisy gate. In this section, we go beyond this generic approach by using a property of the implementable gate set that we call \emph{Pauli-$Z$ compatibility} to modify the Monte Carlo sampling performed in the PEC procedure. 

This yields a lower sampling cost --- i.e. the number of quantum circuits to run --- provided some classical pre-computations can be performed efficiently. We call this improved method \emph{block probabilistic error cancellation} (Block-PEC), as it rearranges the sampling of noisy gates on a per-block basis rather than layer by layer. 

After defining Pauli-$Z$ compatibility, we detail how to implement Block-PEC, how to efficiently (classically) compute the coefficients of its quasi-probabilistic distribution, and extend the technique to consider generic gate sets which contain gates which are not Pauli-$Z$ compatible.

\subsection{Pauli-$Z$ compatibility}
\label{sub:pauli-z-compatibility}

In order to capture a useful commutation property between implementable gate sets and noise, we define the following property.

\begin{definition}[Pauli-$Z$ compatibility]
Let $\ops N$, $\ops U$ and $\ops V$ be three sets of unitary operators acting on at most $n$ qubits, where, 
\begin{enumerate}
  \item[(P1)] $\ops V$ is a group, and a subset of $\ops U$,
  \item[(P2)] $\ops V$ is stable by conjugation with elements of $\ops U$, i.e.
\begin{equation}
  \forall \pop V \in \ops V, \forall \pop U \in \ops U, \exists \pop V' \in \ops V, s.t. \pop U \circ \pop V = \pop V' \circ \pop U;
  \end{equation}
\item[(P3)] $\ops V$ commutes with all $\ops N$, i.e.
\begin{equation}
  \forall \pop V \in \ops V, \forall \pop N \in \ops N,  \pop N \circ \pop V = \pop V \circ \pop N.
  \end{equation}
\end{enumerate}
In such case, $(\ops U, \ops V)$ is said to be $\ops N$-\emph{compatible}. When $\ops N$ is the set of convex combinations of unitary channels made of \emph{Pauli-$Z$ strings} --- i.e. $n$-fold tensor product of single qubit $I$ and $Z$ --- we say that $(\ops U,\ops V)$ is \emph{Pauli-$Z$ compatible}.

For convenience, we will say that a gate is Pauli-$Z$ compatible if is in a set $\ops U$ such that $(\ops U, \ops V)$ is Pauli-$Z$ compatible. A circuit is Pauli-$Z$ compatible if it is generated from a Pauli-$Z$ compatible gate set.
\end{definition}

In this definition, $\ops U$ above represents perfect implementable gates for a given architecture. (P1) states that it is possible to single-out a subgroup $\ops V$ of $\ops U$ that will be used to apply additional controls after each intended gate. (P2) states that given a sequence of implementable gates, one can apply the additional controls described above at the very end of the circuit. 

To make this concrete imagine a 2-gate circuit $\pop U_2 \circ \pop U_1$ with $\pop U_1, \pop U_2 \in \ops U$, and that one wants to implement instead $\pop V_2 \circ \pop U_2 \circ \pop V_1 \circ \pop U_1$, with $\pop V_1, \pop V_2 \in \ops V$. Then, using (P2), there exists $\pop V'_1$ such that $\pop V_2 \circ \pop U_2 \circ \pop V_1 \circ \pop U_1 = \pop V_2 \circ \pop V'_1 \circ \pop U_2 \circ \pop U_1$. Using (P1), $\pop V_2 \circ \pop V'_1$ can be replaced by $\pop V'\in \ops V$. Finally (P3) guarantees that, in spite of noise operators chosen in $\ops N$, the previous pulling of additional controls at the end of the circuit can be performed without modification. For this, let $\pop N_1$ and $\pop N_2$ in $\ops N$ be the noise channels applied after $\pop U_1$ and $\pop U_2$ respectively, then (P3) gives $ \pop V_2 \circ \pop N_2 \circ \pop U_2 \circ \pop V_1 \circ \pop N_1 \circ \pop U_1 = \pop V' \circ \pop N_2 \circ \pop U_2 \circ \pop N_1 \circ \pop U_1$, with the same $\pop V'$ as in the noiseless case. 

With cat-qubit architectures in mind, we can already develop an example of Pauli-$Z$ compatibility. For instance, take $\ops U = \ops V \cup \{\pop{Z}(\theta), \pop{ZZ}(\theta), \pop{CNOT}\}$, where $Z(\theta) = \exp\left(- i \frac{\theta}{2} Z\right)$, $ZZ(\theta) = \exp\left(- i \frac{\theta}{2} Z \otimes Z\right)$, and $\ops V$ are all Pauli-$Z$ strings. Then, a straightforward calculation leads to establishing the Pauli-$Z$ compatibility of $(\ops U, \ops V)$.

Moreover, the noise model associated to $\ops N$ is dephasing, and therefore, the physical intuition behind Pauli-$Z$ compatibility is that any circuit made of gates in $\ops U$ can be sampled from as a block, by performing all the additional controls of the error cancellation procedure at the end of the circuit, rather than after each gate.

Note that if one instead restricts $\ops U$ to a subset of the Clifford group such that it contains the Pauli group, then for $\ops V$ the Pauli group and $\ops N$ the set of Pauli channels, properties P1 to P3 above are satisfied so that it can be said that $(\ops U, \ops V)$ is \emph{Pauli-compatible}.

\begin{figure}[ht!]
    \centering
    \begin{tabular}{c}
        \includegraphics[width=0.8\linewidth]{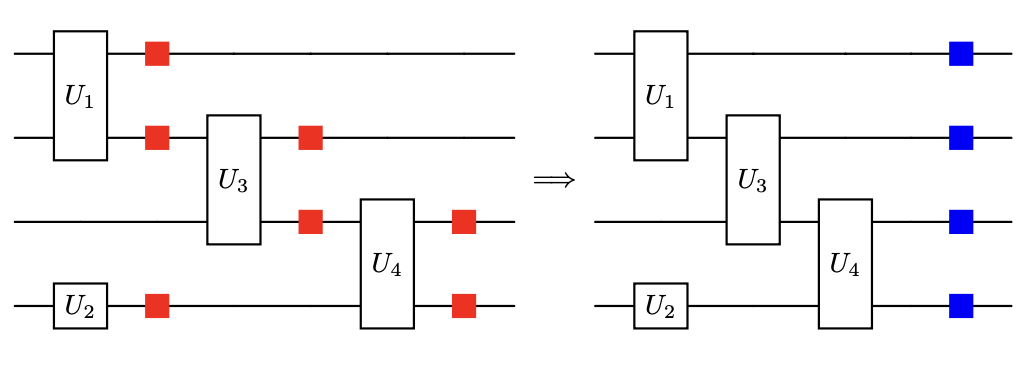}
    \end{tabular}
    \caption{As shown in the example depicted above, Block-PEC samples from the circuit on the left as a single block of the equivalent circuit on the right, where all of the error locations (in red) have been commuted to the right of the circuit and aggregated (in blue).}
    \label{tab:std-to-blk-pec}
\end{figure}

\subsection{Block probabilistic error cancellation}
\label{sub:block-pec-non-uniform}

Consider $(\ops U, \ops V)$ and $\ops N$ satisfying P1 to P3 above, and a layered circuit $\pop U=\pop U_d \circ \cdots \circ \pop U_1$ with layers $\pop U_l \in \ops U$, each associated to their own noise channel $\pop N_l$, forming a non-uniform noise model. Each layer $\pop U_l$ can be decomposed into a  quasi-probabilistic distribution of layers obtained by applying a perfect additional control:
\begin{equation}
\pop U_l = \sum_{\pop V_l \in \ops V} \alpha_{\pop U_l}(\pop V_l) \pop V_l \circ \nop U_l,
\end{equation}
where $\alpha_{\pop U_l}(\pop V_l)$ are real numbers that depend on the target transformation $\pop U_l$ and the effective noise $\pop N_l$ at layer $\pop U_l$. 

We can now commute all the additional controls $\pop V_l$ to the end of the circuit. This results  into a quasi-probabilistic distribution corresponding to $\pop U$:
\begin{align}
\pop U
& = \pop U_d \circ \cdots \circ \pop U_1 \nonumber \\
& = \sum_{\pop V_1, \cdots \pop V_d \in \ops V} \Bigr[\alpha_{\pop U_d}(\pop V_d) \cdots \alpha_{\pop U_1}(\pop V_1)\times \nonumber \\
& \qquad \quad \left.\left(\pop V'_d \circ \cdots \circ \pop V'_1 \right) \circ \left(\nop U_d \circ \cdots \circ \nop U_1\right)\right] \nonumber\\
& = \sum_{\pop V \in \ops V} \alpha_{\pop U}(\pop V') \pop V' \circ \nop U, \label{eq:qp_w_perfect_control}
\end{align}
where the last equality uses the group structure of $\ops V$, so that $\pop V'$ is associated to a depth-1 circuit formed by the product $\prod_{l=1}^d \pop V'_l$ for $\pop V'_l$ defined as the result of commuting $\pop V_l$ through the noise channels $\pop N_k$ for $k\geq l$ and the gates $\pop U_k$ for $k > l$, and where 
\begin{align}
    \alpha_{\pop U}(\pop V') = \sum_{\pop V_1, \cdots \pop V_d \in \ops V} \alpha_{\pop U_d}(\pop V_d)\cdots \alpha_{\pop U_1}(\pop V_1),
\end{align}
for $\pop V_1, \cdots \pop V_d$ such that $\pop V'_{d}\cdots \pop V'_{1} = \pop V'$.

Now, we can apply a similar decomposition for $\pop V' = \sum_{\pop W} \beta_{\pop V'}(\pop W)\nop W$ so that when replaced in Eq.\,\ref{eq:qp_w_perfect_control}, we obtain a new quasi-probabilistic decomposition for $\pop U$, using only noisy implementable gates:
\begin{equation}
\pop U = \sum_{\pop W \in \ops V} \delta_{\pop U}(\pop W) \nop{W} \circ \nop U, \label{eq:qp_w_noisy_control}
\end{equation}
with $\delta_{\pop U}(\pop W) = \sum_{\pop V' \in \ops V} \alpha_{\pop U}(\pop V')\beta_{\pop V'}(\pop W)$.

The interest is that instead of requiring a sample for each layer of the circuit, only a single layer needs to be sampled as long as the new coefficients $\delta_{\pop U}(\pop W)$ can be efficiently computed classically as to ensure that sampling from this distribution can be done efficiently.

To understand why our per-block approach can reduce the sampling overhead, consider a limiting case where the additional controls $\pop V_i$ as well as $\pop V$ can be applied without introducing additional noise. The standard probabilistic error cancellation gives:
\begin{align}
  \pop U = \sum_{\pop V_1, \cdots, \pop V_d \in \ops V} \alpha_{\pop U_1}(\pop V_1) \cdots \alpha_{\pop U_d}(\pop V_d) \pop V_d \circ \nop U_d \circ \cdots \circ \pop V_1 \circ \nop U_1,
\end{align}
which in turn requires a sampling overhead that scales as the square of 
\begin{align}
    \gamma_{\text{std}} \coloneqq \sum_{\pop V_1, \ldots \pop V_d \in \ops V} |\alpha_{\pop U_1}(\pop V_1) \cdots \alpha_{\pop U_d}(\pop V_d)|, \label{eq:gamma_std}
\end{align}
see \cite{zne-pec}.

For our per-block method this is instead equal to 
\begin{align}
\label{eq:block-pec-gamma}
    \gamma_{\text{blk}} \coloneqq \sum_{\pop V' \in \ops V} \left | \sum_{\substack{\pop V_1, \ldots, \pop V_d \in \ops V \\ \pop V'_d \circ \cdots \circ \pop V'_1 = \pop V'}} \alpha_{\pop U_1}(\pop V_1) \cdots \alpha_{\pop U_d}(\pop V_d)\right|. 
\end{align}
The inequality $\gamma_{\text{std}} \geq  \gamma_{\text{blk}}$ follows from the triangle inequality.

Equations~\ref{eq:gamma_std} and \ref{eq:block-pec-gamma} provide two insights in the structure of the advantage one can gain by applying Block-PEC over standard PEC. First, whenever it can be applied, it is never detrimental to the sampling overhead. Second, and most importantly, bigger gains will be obtained whenever, for a given $\pop V' \in \ops V$, the sum of the contributing $\alpha_{\pop U_1}(\pop V_1) \cdots \alpha_{\pop U_d}(\pop V_d)$ terms to $\pop V'$ in Equation~\ref{eq:block-pec-gamma} is small compared to the sum of their absolute values in Equation~\ref{eq:gamma_std}. This requires having terms with different signs, and as there are a number of terms that grows exponentially with $d$, one can expect up to an exponential gain in $d$ by using Block-PEC. We will show in Section~\ref{sec:guarantees} that this is the case for broad classes of circuits and realistic noise models. Section~\ref{sec:applications} explore the gain on end-to-end applications, some of which can run entirely on bias-preserving gates, adding to the interest of Block-PEC for cat-qubits.

We will show in Section~\ref{sec:guarantees} that this is the case for broad classes of circuits and realistic noise models with bias-preserving gates. Section~\ref{sec:applications} explores the gain on end-to-end applications, some of which can run entirely on bias-preserving gates, adding to the interest of Block-PEC for cat-qubits.

Note that the described gain is not restricted to cat-qubits: Similar results can be obtained for Pauli-compatible circuits---i.e. Clifford circuits with Pauli controls under Pauli noise. This regime is explored in~\cite{SHH24reducing} that appeared after the completion of this work.

\subsection{Computing Block-PEC coefficients}
\label{sub:block-pec-coefficients}

Calculating naively the coefficients of Block-PEC as sums of products of coefficients of conventional PEC as presented in Equation~\ref{eq:block-pec-gamma} can take up to $2^{2n(d-1)}$ operations, for a $d$-layer circuit with $d \geq 2$. 

This classical pre-processing cost can be drastically reduced by calculating Block-PEC coefficients recursively, building up on the calculations of the coefficients in sub-circuits.

Indeed, one can observe that any per-block coefficient can be rewritten as follows:
\begin{align}
  \alpha_{\pop U}(\pop V')
    & = \sum_{\substack{\pop V_1, \ldots \pop V_d \in \ops V \\ \pop V'_d \circ \cdots \circ \pop V'_1 = \pop V'}} \alpha_{\pop U_1}(\pop V_1) \cdots \alpha_{\pop U_{d}}(\pop V_{d})\\
    & = \sum_{\substack{\pop V_{d} \in \ops V}} 
    \alpha_{\pop U_{d}}(\pop V_{d}) \times \nonumber \\
    & \quad
    \sum_{\substack{\pop V_1, \ldots, \pop V_d \in \ops V \\ \pop V'_{d} \circ \cdots \circ \pop V'_1 = \pop V' }}
    \alpha_{\pop U_1}(\pop V_1) \cdots \alpha_{\pop U_{d-1}}(\pop V_{d-1})\\
    & = \sum_{\substack{\pop V_{d} \in \ops V}} 
    \alpha_{\pop U_{d}}(\pop V_{d})
         \alpha_{\pop U_{d-1} \circ \cdots \circ \pop U_1}(\pop W'),
\end{align}
where $\pop W' = \pop V_d^{-1} \circ \pop V'$ should be thought of as the result of commuting $\pop V_{1}$ to $\pop V_{d-1}$ through all gates and noise channels to the end of the circuit until  reaching, but not crossing, the $d$-th layer.

This opens the door to a recursive calculation of per-block coefficients, for the sub-circuit $\pop U_{l+1}\pop U_l\cdots\pop U_1$, from the per-block coefficients for the sub-circuit  $\pop U_{l}\cdots\pop U_1$.

\begin{algorithm}
    \caption{Block-PEC coefficient calculation}
    \begin{algorithmic}
    \Function{compute-coefficients}{$\pop U_{l+1} \circ \cdots \circ \pop U_1$, $\pop V'$}
    \If{$l = 1$}
        \State \Return $\textbf{list}(\alpha_{\pop U_2}(\pop V_2) \alpha_{\pop U_{1}}(\pop V_{1}) \textbf{ for } \pop V_2,\pop V_1 \in \ops V^2)$
    \EndIf
    \ForAll{$\pop V_{l+1} \in \ops V$} 
        \State $\beta(\pop V_{l+1}) \gets 0$
        \ForAll{$\pop W \in \ops V$} 
            \If{$\pop V'_d \circ \pop W' = \pop V'$} 
                \State $\beta(\pop V_{l+1}) \gets \beta(\pop V_{l+1})
                + \alpha_{\pop U_{d}}(\pop V_{d}) \cdot 
                \Call{compute-coefficients}{\pop U_l \circ \cdots \circ \pop U_1,\pop W'}$
            \EndIf    
        \EndFor
    \EndFor
    \State \Return $\textbf{list}(\beta(\pop V_{l+1}) \textbf{ for } \pop V_{l+1} \in \ops V)$
    \EndFunction
    \end{algorithmic}
\end{algorithm}

For any $n$-qubit $\pop V \in \ops V$, where $\ops V$ is the set of Pauli-$Z$ string, calculating $\alpha_{\pop U_2 \pop U_1}(\pop V)$ takes at most $|\ops V|^{2} \leq 2^{2n}$ operations. 
Similarly, calculating each $\alpha_{\pop U_{l+1} \pop U_l\cdots \pop U_1}(\pop V)$ knowing $\{\alpha_{\pop U_l\cdots \pop U_1}(\pop V)\}_{\pop V}$ takes at most $|\ops V|^{2} \leq 2^{2n}$ operations. 
Starting from $2^n$ coefficients, one needs to perform $2^n$ multiplications and additions to compute all the per-block coefficients $\{\alpha_{\pop U_{l+1} \pop U_l\cdots \pop U_1}(\pop V)\}_\pop V$.

All in all, calculating recursively all per-block coefficients of a $d$-layered circuit takes at most $(d-1)2^{2n}$ operations, instead of $2^{2n(d-1)}$ for the naive approach.

The effect of Block-PEC can be seen as replacing the relatively easy sampling method of the conventional PEC approach with a more complex sampling which extracts some classically computable quantities and in turn reduces greatly the number of quantum samples required for a given precision.

\subsection{Hybrid Block-PEC for circuits with generic gates}
\label{sub:hybrid-block-pec}

We have previously established that $(\ops U, \ops V)$ is Pauli-$Z$ compatible, where $\ops V$ is the set of all Pauli-$Z$ string, and $\ops U = \ops V \cup \{\pop{Z}(\theta), \pop{ZZ}(\theta), \pop{CNOT}\}$. Both the $T$ gate ($T = Z(\pi / 4)$) and the phase gate ($S = Z(\pi / 2)$) are in $\ops U$.  The Hadamard gate $H$ together with the set $\ops U$ form a universal gate set, but $(\ops U \cup \{H\},\ops V)$ is not Pauli-$Z$ compatible. 

Nonetheless, our per-block technique can be applied to Pauli-$Z$ compatible sub-circuits.  This results in a hybrid Block-PEC / PEC simulation of the noiseless circuit using samples obtained from the noisy quantum computer. The advantage given by the Block-PEC technique remains, as the lower sampling cost of each sub-circuit directly reduces the sampling cost of the full circuit.

In this context, it would be beneficial to implement circuit rewriting rules which optimize quantum circuit layouts for our error mitigation schemes for example, by optimizing the positions of Hadamard gates within the circuit or to move Hadamard gates along the wires of the circuit. We leave this for future research.

\section{Performance analysis of Block-PEC}
\label{sec:guarantees}

Section~\ref{sec:block-pec} introduced the possibility of having up to an exponential gain in using Block-PEC as the number $d$ of consecutive layers of bias-preserving gates grows. Yet, to make such statement precise, it is necessary to define a gate-set for the bias-preserving gates, noise models and most importantly a class of circuits of interest for which to conduct the analysis of the sampling gain. Section~\ref{sub:performance-analysis} performs this analysis by first considering a \emph{single} bias-preserving gate giving analytic performance guarantees, and then extending the analysis for \emph{several} layers of bias-preserving circuits. Two cases are studied here, relatively shallow circuits with random arrangements of bias-preserving gates, and deeper circuits corresponding to $SWAP$ networks~\cite{OHRW19generalized} that are used for superconducting NISQ architectures. Section~\ref{sub:partial} extends the latter to a wide class of circuits that can be directly used to construct applications. The impact on end-to-end applications is considered in  Section~\ref{sec:applications}.

\subsection{Strictly bias-preserving gates}
\label{sub:performance-analysis}
\paragraph{Pauli-$Z$ compatibility.}

A gate is said to be bias-preserving if it does not turn phase-flip errors into bit-flip errors. In other words, phase-flip errors are mapped to phase-flip errors, preserving phase-flip biases (see Appendix~\ref{sec:bias-preserving} for a complete mathematical characterization). 

\begin{definition}
A $n$-qubit gate is bias-preserving if its unitary matrix representation $G$ is such that for every $n$-by-$n$ unitary diagonal matrix $D$, there is a $n$-by-$n$ unitary diagonal matrix $D'$ such that $GD = D'G$.
\end{definition}

Examples of bias-preserving gates are the gates $X$, $Z$, $CNOT$, $CZ$, the $Z$-axis rotation gates $Z(\theta)$ and $ZZ(\theta)$, and the Toffoli gate, along with any gate which can be implemented by a sequence of bias-preserving gates. 
In what follows, we determine which bias-preserving gates are Pauli-$Z$ compatible.

\begin{proposition}
\label{prop:bp-pauli-z}
    The gates X, CZ, $Z(\theta)$, $ZZ(\theta)$, $CNOT$ and $SWAP$ are Pauli-$Z$ compatible.
    The Toffoli gate is not Pauli-$Z$ compatible.
\end{proposition}

\begin{proof}
The unitaries $X_{[i]}$, $CZ_{(i,j)}$ (with target qubit $j$), the Z-axis rotation unitaries $Z(\theta)$ and $ZZ(\theta)$, and any unitary $Z_{[A]}$ (for $A \subseteq \{1,\ldots,n\}$), are bias-preserving unitary gates $U$ such that $Z_iU=UZ_i$ for every index $i$, and therefore they are Pauli-$Z$ compatible. 

Additionally, the $CNOT_{(i,j)}$ gate (where we take $j$ to be the target qubit) is such that 
\begin{align}
    Z_{[i]} CNOT_{(i,j)} &= CNOT_{(i,j)} Z_{[i]},\\
    \text{and }
    Z_{[i,j]} CNOT_{(i,j)} &= CNOT_{(i,j)} Z_{[j]},
\end{align}
as shown in Figure~\ref{tab:cnot-propagation}. Therefore, both the $CNOT$ gate and the $SWAP$ gate (when decomposed as 3 $CNOT$ gates) are Pauli-$Z$ compatible. 

\begin{figure}[ht!]
    \centering
    \begin{tabular}{cc}
        \includegraphics[width=0.3\linewidth]{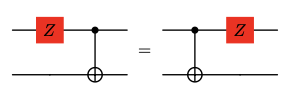} 
      & \includegraphics[width=0.3\linewidth]{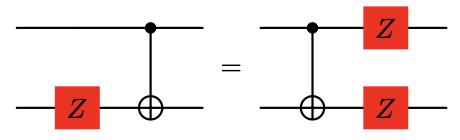} \\
      (a) & (b)
    \end{tabular}
    \caption{Phase-flip propagation through the $CNOT$ gate; the equality is between the unitaries associated to the circuits. In (a) the phase-flip commutes as it affects the control of the $CNOT$ while in (b) it propagates from the target to the control.}
    \label{tab:cnot-propagation}
\end{figure}
However, this is not the case for the Toffoli gate $CCX_{[i,j,k]}$ (with target qubit $k$), which is such that
\begin{align}
    Z_{[i]} CCX_{(i,j,k)} &= CCX_{[i,j,k]} Z_{[i]},\\
    Z_{[j]} CCX_{(i,j,k)} &= CCX_{[i,j,k]} Z_{[j]},\\
    \text{and }
    CZ_{(i,j)} Z_{[k]} CCX_{(i,j,k)} &= CCX_{(i,j,k)} Z_{[k]},
\end{align}
and is therefore not Pauli-$Z$ compatible.
\end{proof}

Hence, for circuits with Toffoli gates, it will be necessary to resort to Hybrid Block-PEC described earlier. Nonetheless, the gates listed in Proposition~\ref{prop:bp-pauli-z} form a generating set for a rather large class of (sub)-circuits on which one can implement Block-PEC.

\paragraph{Analytic performance guarantees.}

Additionally, we studied analytically all combinations of two bias-preserving gates acting on at most 2 qubits, subjected to uncorrelated dephasing noise.

We found out that the per-block method leads to an improvement when a $CNOT$ follows a gate which acts non-trivially on the $CNOT$'s target qubit. In such case, the block would be limited to the single $CNOT$ and the additional $Z$-controls of the preceeding gate.

\begin{proposition}
    Assuming a non-zero 1-qubit phase flip error, Block-PEC yields a sampling cost reduction over PEC, on Pauli-$Z$ compatible circuits with any sub-circuit of the form:
    \begin{enumerate}[label=(\alph*)]
    \item\label{it:block-pec-pattern1} $U_{(a)}=CNOT_{j,k}U_k$;
    \item\label{it:block-pec-pattern2} $U_{(b)}=CNOT_{j,k}U_{j,k}$;
    \item\label{it:block-pec-pattern3} $U_{(c)}=CNOT_{i,j}U_{j,k}$ (with $i \neq k$).
\end{enumerate}
\end{proposition}

\begin{proof}
Define $p$ to be the $1$-qubit phase flip error, and assume that $p > 0$. Let us write $\gamma_{\text{std}}(U)$ (resp. $\gamma_{\text{blk}}(U)$) for the sampling cost associated with the PEC (resp. Block-PEC) protocol for the circuit $U$. We represent these circuit patterns with their error locations in Figure~\ref{tab:circuit-patterns}. We refer to Appendix~\ref{appendix:dephasing} for a definition of dephasing noise channels as linear combinations of Pauli-$Z$ channels, and to Appendix~\ref{appendix:detailed-performance-analysis} for the details of the calculations to follow.

\begin{figure}[ht!]
    \centering
    \begin{tabular}{ccc}
      \includegraphics[width=0.3\linewidth]{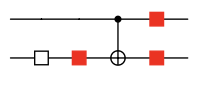}
      & \includegraphics[width=0.3\linewidth]{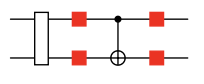}
      & \includegraphics[width=0.3\linewidth]{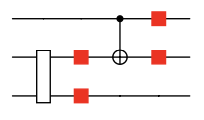}\\
        (a) & (b) & (c)
    \end{tabular}
    \caption{Circuit $U_{(x)}$ with error locations, for $x \in \{a,b,c\}$.}
    \label{tab:circuit-patterns}
\end{figure}

For the circuit $U_{(a)} =CNOT_{j,k} U_k$, one can compute that 
\begin{align}
    \gamma_{\text{blk}}(U_{(a)}) = |\beta_0| + 3 |\beta_1| = \beta_0 - 3 \beta_1,
\end{align}
with $\beta_0 = ((1-p)^3 - p^3) / (1-2p)^3$ and $\beta_1 = (p^2(1-p) - p(1-p)^2) / (1-2p)^3$. Then, 
\begin{align}
    \gamma_{\text{blk}}(U_{(a)}) = \frac{1+2p-2p^2}{(1-2p)^2} < \frac{1}{(1-2p)^3} = \gamma_{\text{std}}(U_{(a)}).
\end{align}
Similarly, one can compute that
\begin{align}
    \gamma_{\text{blk}}(U_{(b)}) &= \frac{1+2p-6p^2+4p^3}{(1-2p)^3}\\ 
    &< \frac{1}{(1-2p)^4} = \gamma_{\text{std}}(U_{(b)}),\\
    \text{and } \gamma_{\text{blk}}(U_{(c)}) &= \frac{1+2p-2p^2}{(1-2p)^3}\\
    &< \frac{1}{(1-2p)^4} = \gamma_{\text{std}}(U_{(c)}).
\end{align}
\end{proof}

This analysis can now be used to understand the possible gains for generic circuits that can be decomposed into sequences of sub-circuits with Pauli-$Z$ compatible ones, interspersed with non Pauli-$Z$ compatible gates. Indeed, whenever the sub-circuits of the type $U_{(a)}, U_{(b)}, U_{(c)}$ appear within a larger circuit, Block-PEC yields an improvement over PEC (see Figure~\ref{tab:pattern-analysis}). The frequency of the appearance of such sub-circuits determines the efficiency of Block-PEC. More concretely, if sub-circuits of type $U_{(a)}, U_{(b)}, U_{(c)}$ appear $k$ times throughout the circuit, the sampling cost decreases by a factor proportional to $kp^2$.

\begin{figure}
    \centering
    \begin{tabular}{c}
        \includegraphics[width=0.45\linewidth]{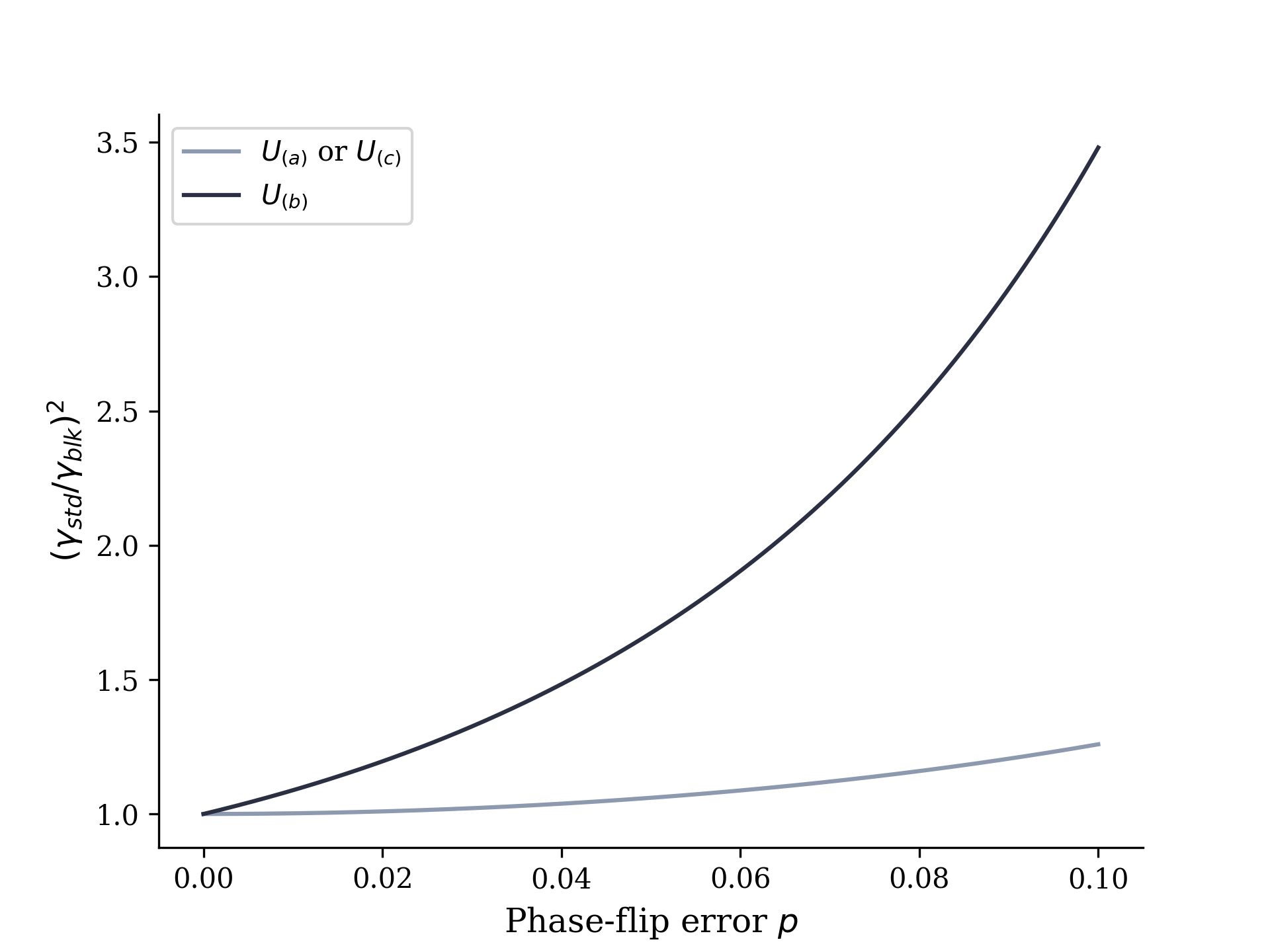}
    \end{tabular}
    \caption{Sampling cost gain  $(\gamma_{\text{std}} / \gamma_{\text{blk}})^2$ for $U_{(x)}$ for $x \in \{a,b,c\}$, as a function of the phase-flip error $p$.}
    \label{tab:pattern-analysis}
\end{figure}

\paragraph{Performance evaluation for random circuits.}

To analyze the performance of Block-PEC for the gate set $\{X, Z, CNOT, Z(\theta), ZZ(\theta), CZ\}$, we have constructed random circuits acting on $n$ qubits and with a maximum depth $n+1$. More precisely, we have sampled uniformly at random $n+1$ labels from the set $\{X, Z, CNOT, Z(\theta), ZZ(\theta), CZ\}$. For each label we have then randomly chosen the qubit(s) the corresponding gate has been acting upon. For $ZZ(\theta)$ gates, each angle $\theta$ has been chosen uniformly at random in $[0,2\pi)$.

We then numerically evaluated the sampling cost gain defined as the ratio $(\gamma_{\text{std}} / \gamma_{\text{blk}})^2$. It shows that for uncorrelated dephasing noise with strength $p=0.1$ and $n=8$, the sampling overhead of PEC is more than twelve times that of Block-PEC, as shown in Figure~\ref{tab:bp-sampling-cost} (left). As expected, whenever the strength of the noise decreases, the difference between the two methods vanishes. This can be seen for $p=0.001$ in Figure~\ref{tab:bp-sampling-cost} (right). Indeed, for low noise strength, the probability of having more than one error per circuit becomes small enough so that the noise is well approximated by a process that would affect a single location in the whole block. This effectively means that sampling per gate, per layer or per block will be equivalent up to negligible factors and hence yield a similar sampling overhead. 

\begin{figure}
    \centering
    \begin{tabular}{cc}
      \includegraphics[width=0.45\linewidth]{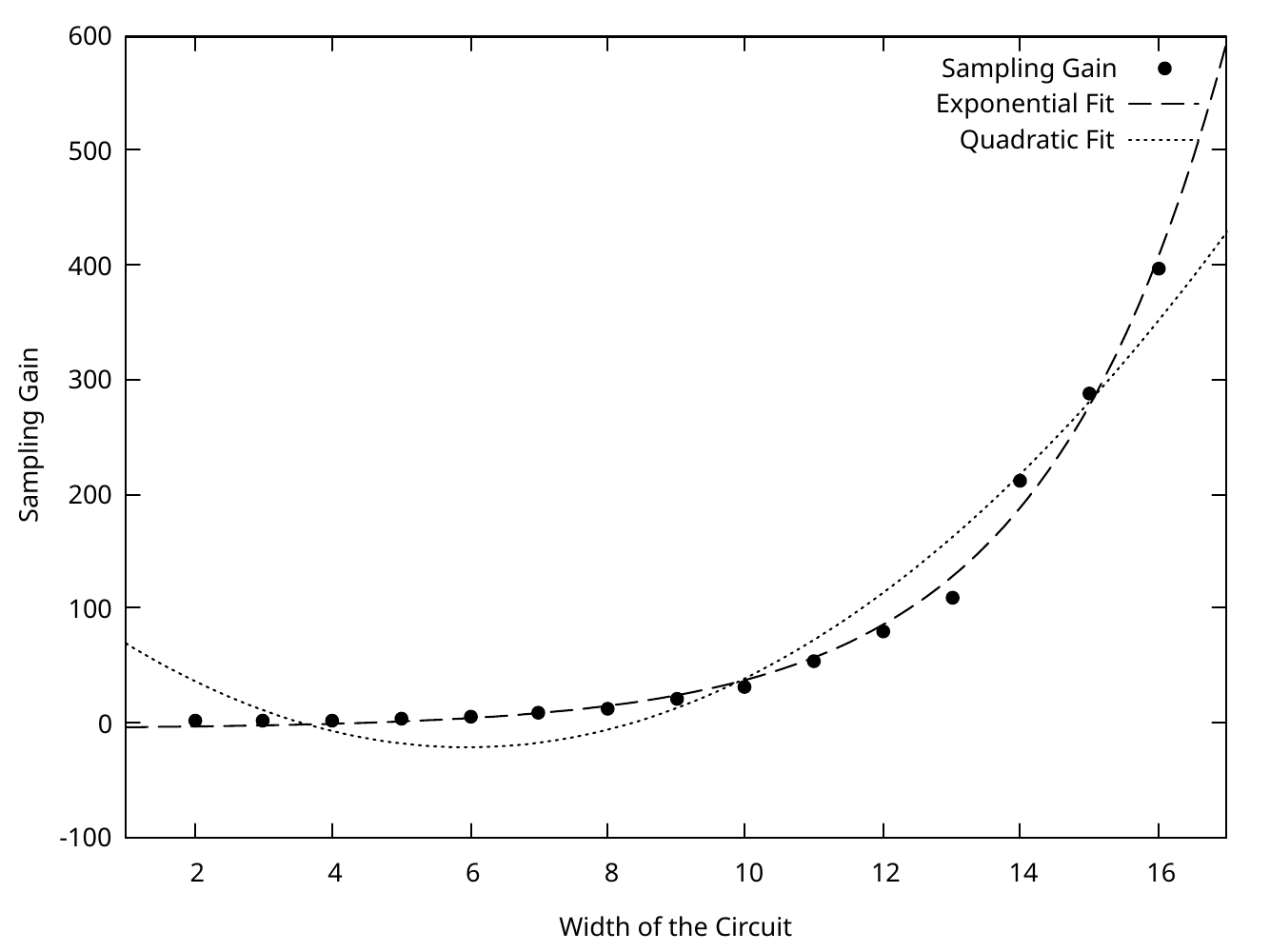} &
      \includegraphics[width=0.45\linewidth]{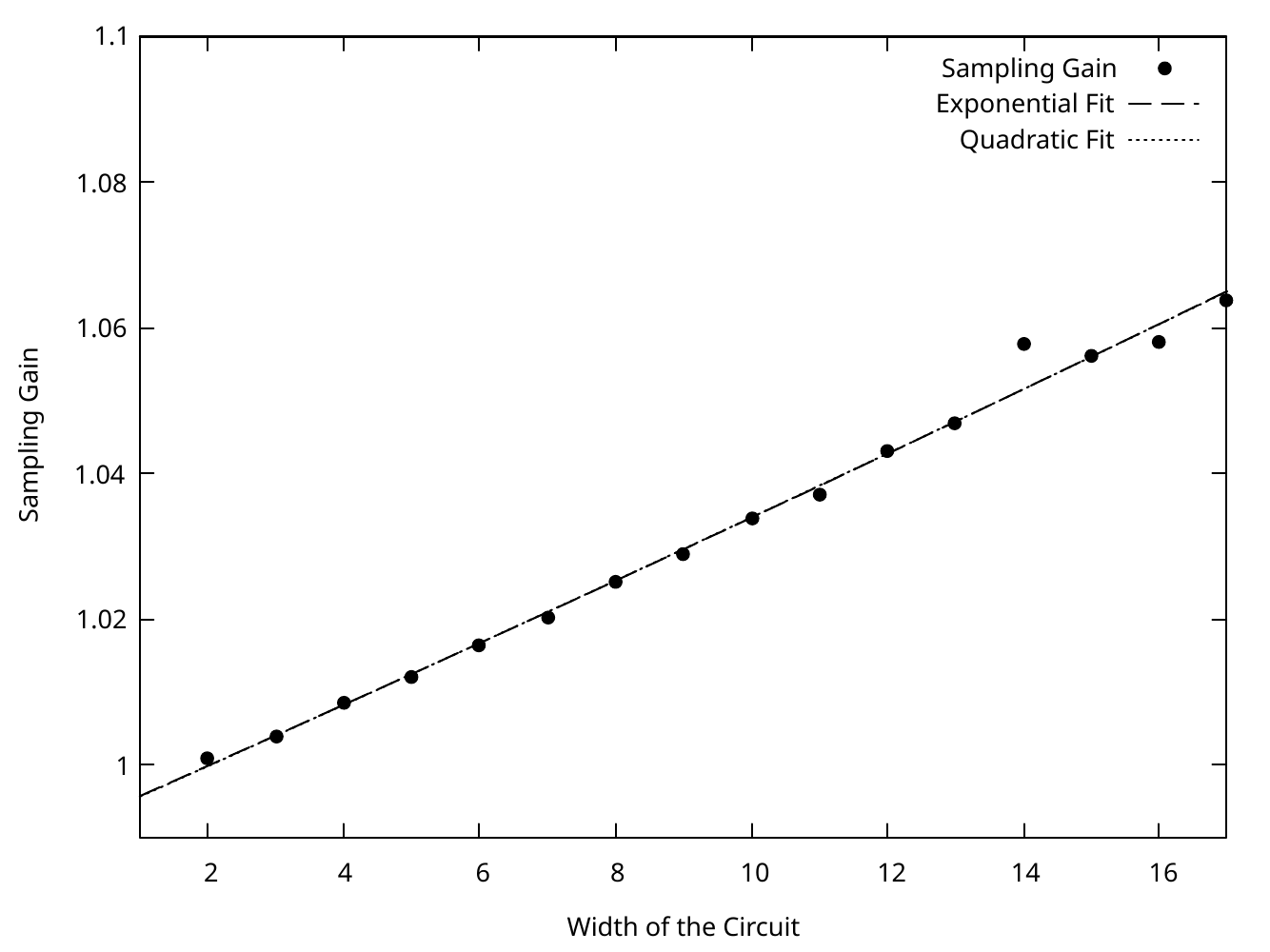}
    \end{tabular}
    \caption{Performance of Block-PEC on randomly generated bias-preserving circuits, for a phasing noise with error probability $p=0.1$ (left) and $p=0.001$ (right). The vertical axis shows the sampling cost gain $(\gamma_{\text{std}} / \gamma_{\text{blk}})^2$ as a function of the qubit count $n$. The depth of the circuit is $n+1$ making them relatively shallow and requiring higher error probabilities to exhibit a sizable gain when using Block-PEC.}
    \label{tab:bp-sampling-cost}
\end{figure}

\paragraph{Performance evaluation for $SWAP$ networks with bias-preserving gates.}

An important class of circuits for superconducting devices in the NISQ regime is that of $SWAP$ networks~\cite{OHRW19generalized}. These circuits allow to optimally swap qubit states so that any two qubit on a line become adjacent during one execution of the network. This then allows to perform a gate between them before putting them apart and allowing for other interactions. In effect, this solves the problem of long range interaction on architectures that have spatially fixed registers, though possibly requiring to repeat the $SWAP$ network to take into account non commuting interaction gates.

Now, when the $SWAP$ gates are compiled into 3 $CNOT$s, and if the gates performing the interactions are bias-preserving, the whole circuit is indeed bias-preserving as well. The depth and nature of the gates with many instances of pattern $U_{(b)}$ (Fig.~\ref{tab:circuit-patterns}-b) hint at the higer interest of implementing Block-PEC for these circuits.

This is demonstrated in Figure~\ref{fig:rzz-swap} where we have considered $SWAP$ networks of increasing width. Their depth was arbitrarily set to one times their width (left) and 3 times their width (right). The interactions before each swap are $ZZ(\theta)$ 2-qubit rotation gates with random $\theta$. There the exponential nature of the gain appears clearly and is confirmed by the total residual squarred error that it smaller for the exponential fit by a factor up to 6.12 when compared to the quadratic fit.

\begin{figure}
  \centering
  \begin{tabular}{cc}
    \includegraphics[width=0.45\linewidth]{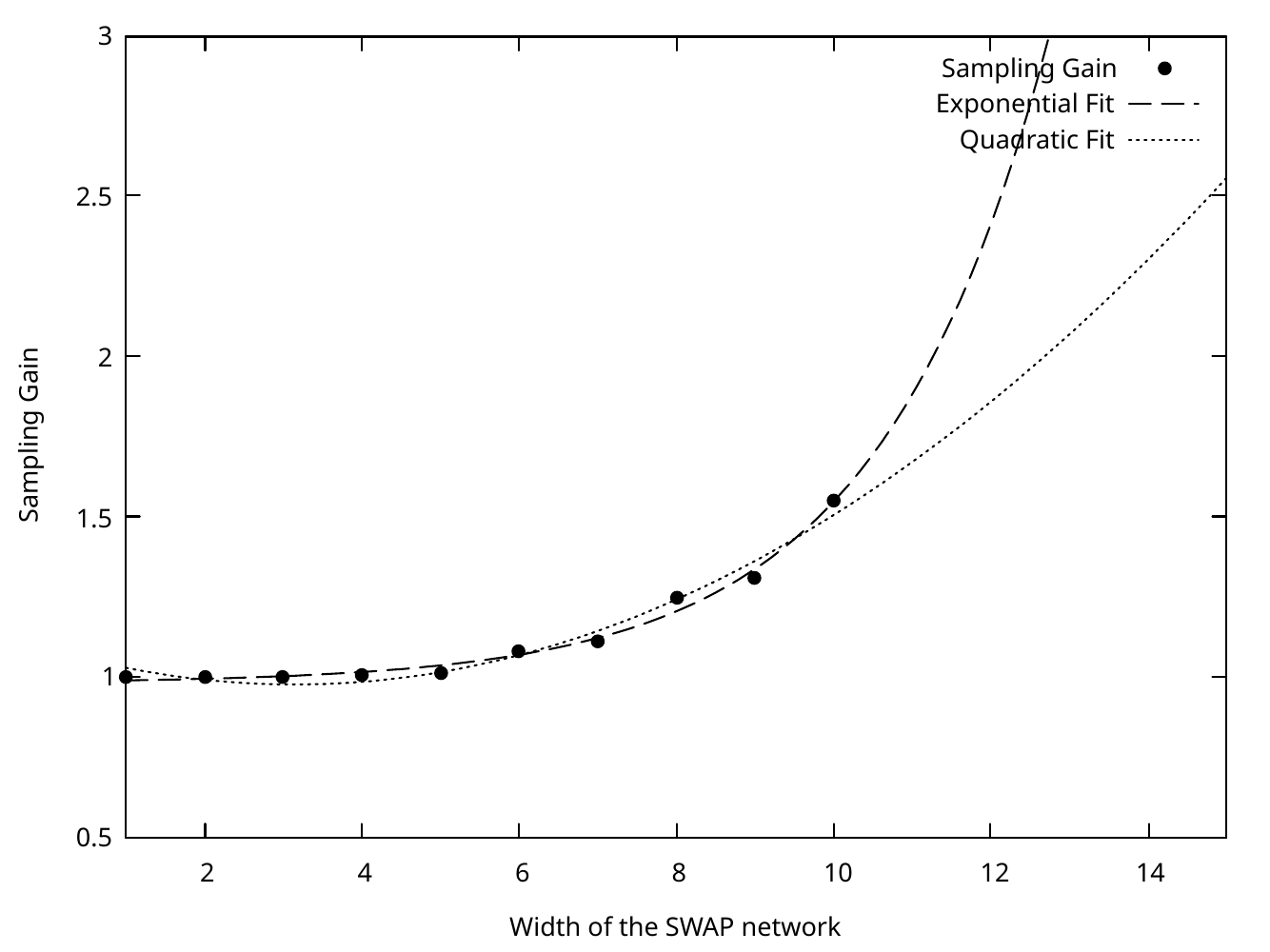}
    & \includegraphics[width=0.45\linewidth]{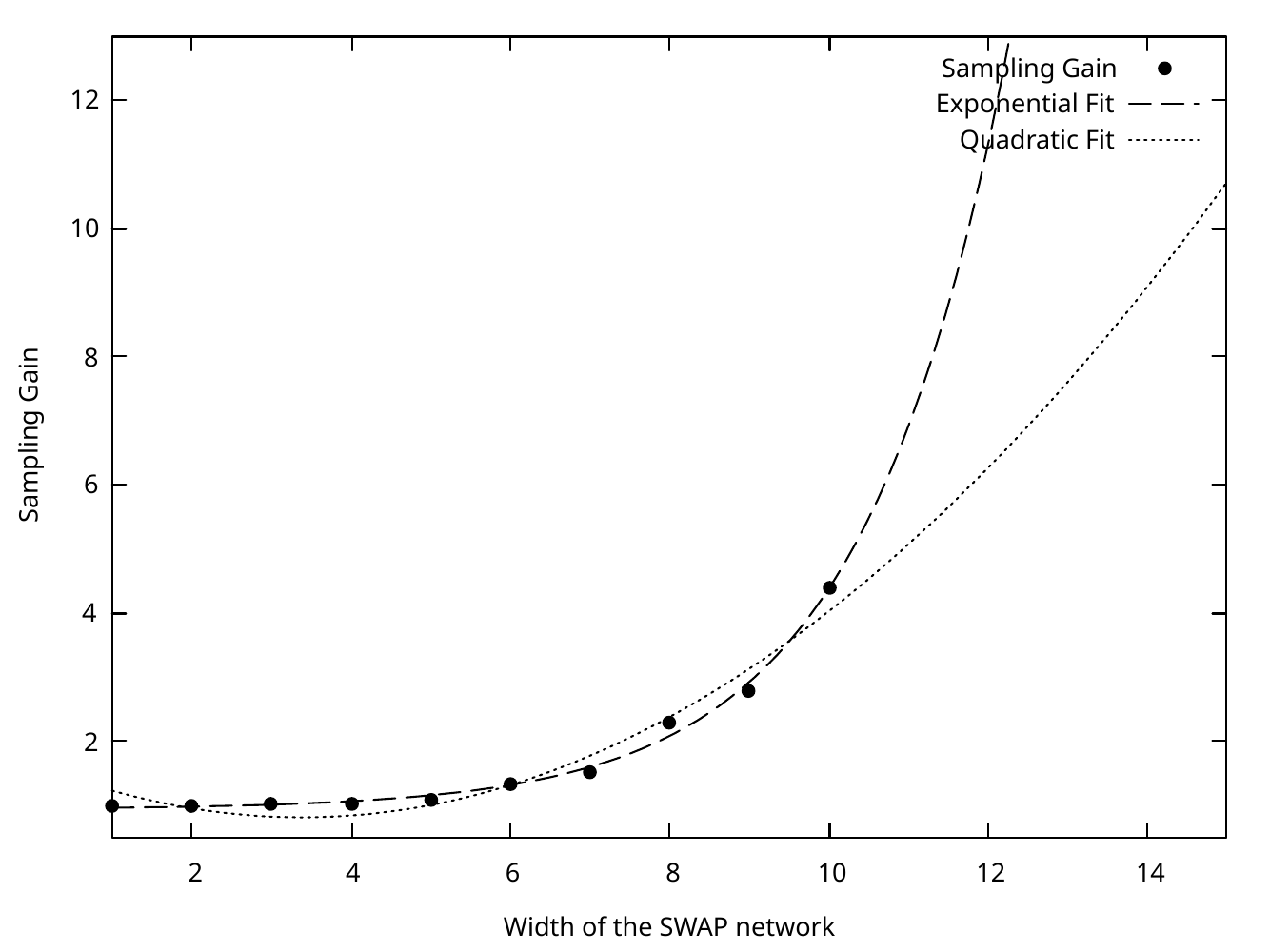}
  \end{tabular}
  \caption[]{Sampling gain for $SWAP$ networks and $ZZ(\theta)$ gates. The depth is equal to the width (left) or 3 times the width (right). The error model is independent uncorrelated single qubit dephasing noise for each wire touched by a gate. The dephasing probability is 0.001. The deeper the circuit, the more advantageous it is to perform Block-PEC, with already a factor 4 for 9 qubits only in the deeper circuit. The total squarred error for the exponential fit are 0.0035 (left) and 0.0756 (right) while it is 0.0083 (left) and 0.4623 (right) for the quadratic one. This entails the exponential fit a respectively 2.38 and 6.12 better precision than the quadratic one.}
  \label{fig:rzz-swap}
\end{figure}

\subsection{Partially bias-preserving circuits}
\label{sub:partial}

\paragraph{Pauli-$Z$  compatibility.}

As the definition of bias-preserving gates is quite restrictive, we have striven to identify gates which preserve phase flip biases but only on a subspace of the space of all quantum states. It is in this context that we studied parity-preserving gates. These gates are parameterized $2$-qubit gates which perform a rotation in the subspace defined by $\text{span}\left\{\ket{01}, \ket{10}\right\}$, while acting as the identity on the rest. We define what it means for such a gate to be partially bias-preserving (see Appendix~\ref{app:partially-bias-preserving} for a complete mathematical characterization).

\begin{definition}
Consider $\mathcal S$ to be a subspace of the Hilbert space of $n$ quantum states.
A gate is $\mathcal S$-bias-preserving if its unitary matrix representation $G$ is:
\begin{enumerate}
    \item stable under $\mathcal S$, i.e. $\ket \phi \in \mathcal S \implies G\ket \phi \in \mathcal S$;
    \item such that for every $n$-by-$n$ unitary diagonal matrix $D$, there is a $n$-by-$n$ unitary diagonal matrix $D'$ such that, for every state $\ket\phi \in \mathcal S$, $GD\ket\phi=D'G\ket{\phi}$.
\end{enumerate}
\end{definition}

RBS gates are $2$-qubit rotation gates acting on $\ket{01}$ and $\ket{10}$ as
\begin{align}
RBS(\theta)\ket{01} &= \cos(\theta)\ket{01} - \sin(\theta)\ket{10}\\
RBS(\theta)\ket{10} &= \cos(\theta)\ket{10} + \sin(\theta)\ket{01}
\end{align}
while acting as the identity on $\ket{00}$ and $\ket{11}$.

We use parity-preserving gates, and in particular $RBS$ gates, as examples of partially bias-preserving gates. Other examples of such partially bias-preserving parity-preserving gates are the $\text{A}(\theta,\phi)$ gate (see e.g.~\cite{Barkoutsos_2018,Gard_2020,vpe-vff}), the $
\text{Givens}(\theta)$ gate (see e.g.~\cite{Anselmetti_2021,Kivlichan_2018,arute2020}), the $\text{XY}(\theta)$ gate (see e.g.~\cite{Abrams_2020}).

\begin{theorem}
    $RBS$ gates are partially bias-preserving gates.
\end{theorem}

\begin{proof}
    See Appendix~\ref{app:partially-bias-preserving}.
\end{proof}

The $RBS$ gate can be implemented as the following sequence of three gates
\begin{align}
\label{eq:rbs-implementation}
    RBS(\theta)_{j,k} = CNOT_{j,k}
    XCZ(\theta)_{j,k}
    CNOT_{j,k},
\end{align}
where $XCZ(\theta)$ is the $X$-controlled $Z(\theta)$ rotation gate, represented by the unitary matrix
\begin{align}
    XCZ(\theta) = \frac{1}{2} 
    \begin{pmatrix}
        1 + e^{i\frac{\theta}{2}} & 0 
        & 1 - e^{i\frac{\theta}{2}} & 0\\
        0 & 1 + e^{-i\frac{\theta}{2}}
        & 0 & 1 - e^{-i\frac{\theta}{2}}\\
        1 - e^{i\frac{\theta}{2}} & 0 
        & 1 + e^{i\frac{\theta}{2}} & 0\\
        0 & 1 - e^{-i\frac{\theta}{2}}
        & 0 & 1 + e^{-i\frac{\theta}{2}}\\
    \end{pmatrix}. \nonumber
\end{align}
Each of these gates is Pauli-$Z$ compatible, making the $3$-gate circuit implementing $RBS$ Pauli-$Z$ compatible.

\paragraph{Performance evaluation for $RBS$ gate circuits.}
The previous decomposition of $RBS$ gates contains the pattern $U_{(b)}$ as $XCZ(\theta)_{j,k} CNOT_{j,k}$.  Therefore, as per the analysis of the previous subsection, one can expect non-trivial gains when using Block-PEC on algorithms which sample from $RBS$ gate-based parameterized circuits.

We have numerically evaluated the pyramid circuit represented in Figure~\ref{tab:matchgate-architectures} as it is a useful building block of several quantum machine learning algorithms (See Section~\ref{sec:applications}). Our simulations (Figure~\ref{tab:rbs-numerics}) show that Block-PEC exponentially reduces the number of samples necessary to achieve a given precision for the error mitigated value for circuits as the depth of the circuit increases. For instance, extrapolating our numerical simulations with an exponential fit let us expect a 12 times sampling gain for 30 qubits at $p=0.001$.

\begin{figure}
    \centering
    \begin{tabular}{c}
        \includegraphics[width=0.8\linewidth]{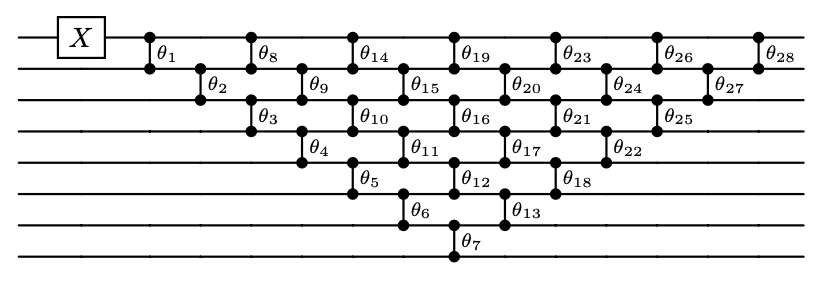}
    \end{tabular}
    \caption{8-qubit pyramidal layering for matchgates.}
    \label{tab:matchgate-architectures}
\end{figure}

\begin{figure}
    \centering
    \includegraphics[width=.5\linewidth]{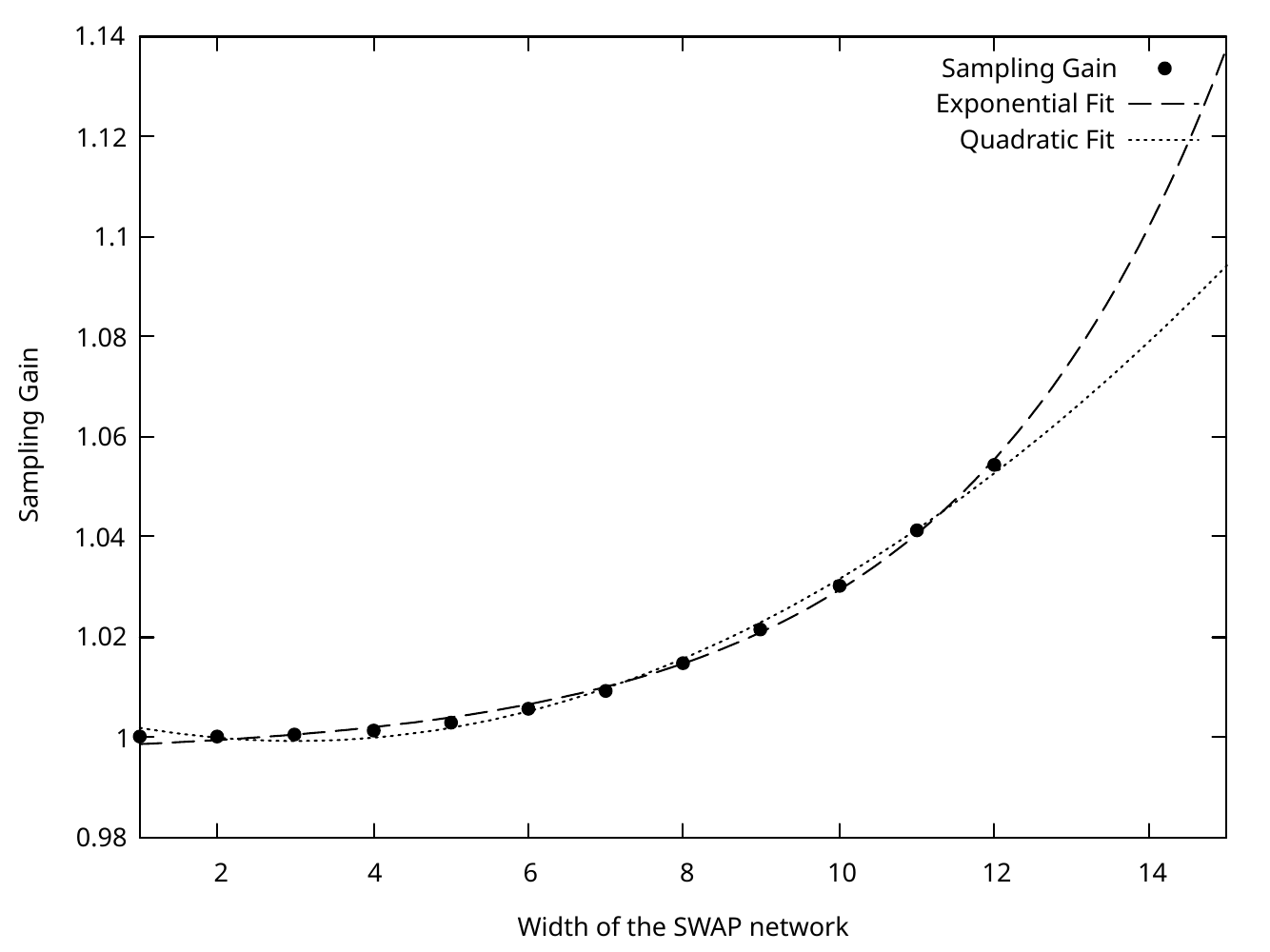}
    \caption{Performance of Block-PEC on pyramids of $RBS$ gates (Figure~\ref{tab:matchgate-architectures}), with a single-qubit phase-flip noise model with error probability $p=0.001$. The exponential fit exhibits a residual squared error of $8.07.10^{-6}$ while the quadratic fit is at $1.646.10^{-5}$. The best fit is $\approx 0.0018 e^{ 0.2905 n} + 0.9962$ where $n$ is the number of qubits in the pyramid circuit. For $n = 30$ the expected sampling gain is about 12 times.}
    \label{tab:rbs-numerics}
\end{figure}

\paragraph{Performance evaluation for $SWAP$ networks with $RBS$ gates.}

As for the case of fully bias-preserving circuits, $SWAP$ networks~\cite{OHRW19generalized} are useful to consider in the case of partially bias-preserving circuits. Here again, when the $SWAP$ gates are compiled into 3 $CNOT$s, with $RBS$ gates for the interaction gates, the whole circuit is partially bias-preserving.

The efficiency of Block-PEC is particularly stricking in Figure~\ref{fig:rbs-swap-net} where we have considered $SWAP$ networks of increasing width. Their depth was arbitrarily set to one times their width (left) and 3 times their width (right). The interactions between the qubits before each swaps are $RBS(\theta)$ gates with random $\theta$. Again the exponential nature of the gain appears clearly and is confirmed by the fact that the exponential fit exhibits a total squarred error smaller than that of the quadratic fit by a factor up to 16.04.

\begin{figure}
  \centering
  \begin{tabular}{cc}
    \includegraphics[width=0.45\linewidth]{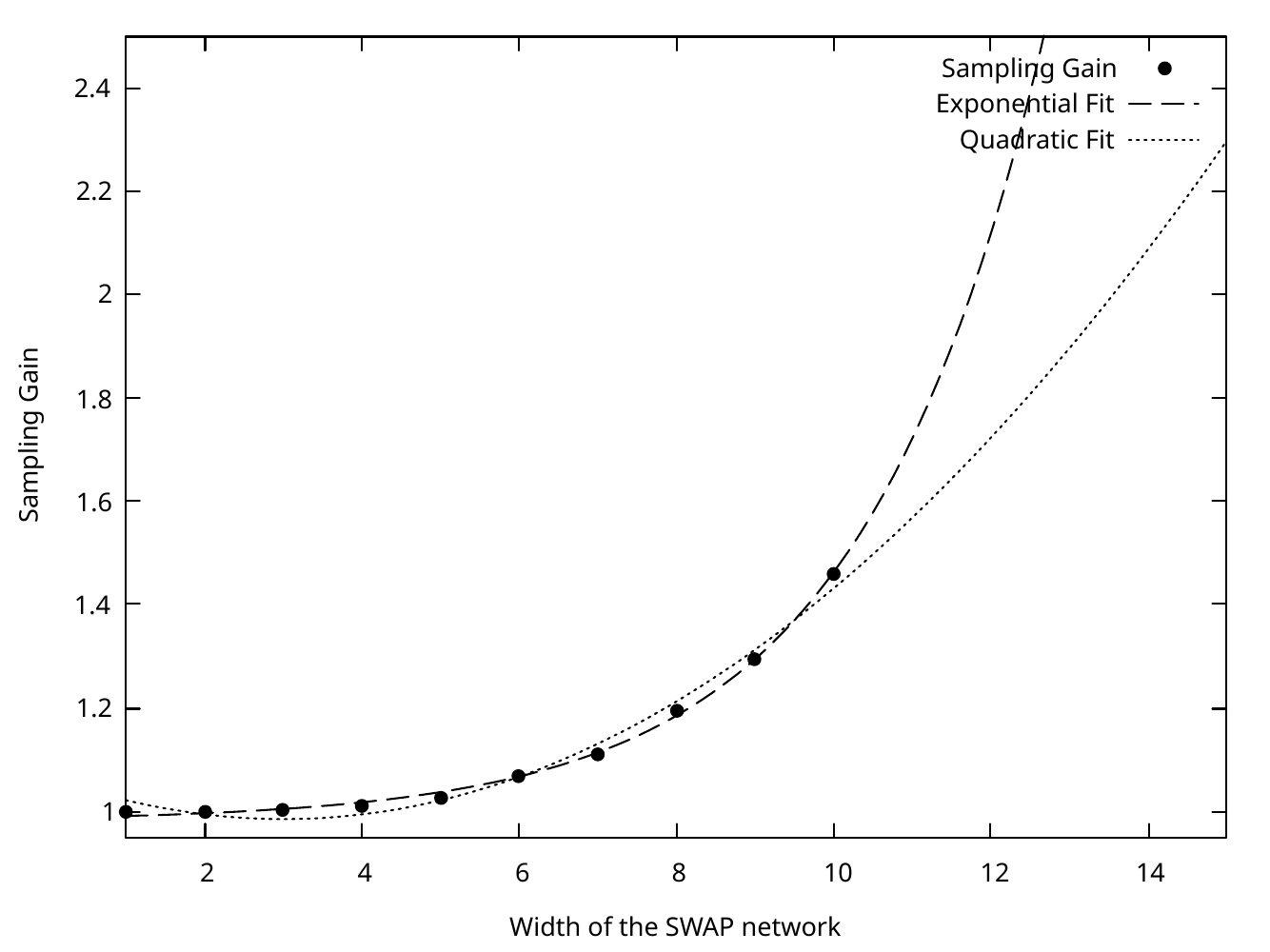}
    & \includegraphics[width=0.45\linewidth]{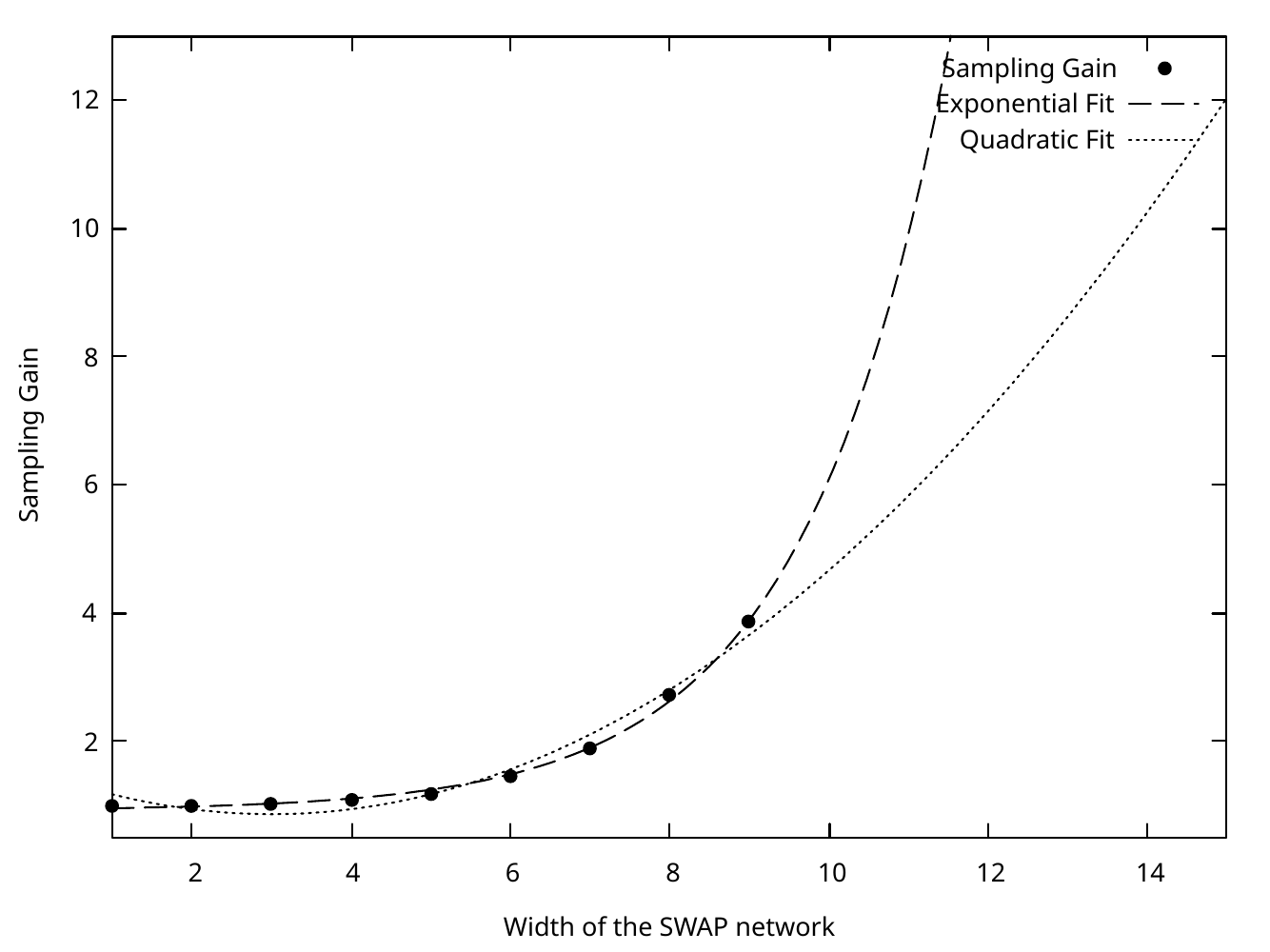}
  \end{tabular}
  \caption[]{Sampling gain for $SWAP$ networks and $RBS$ gates. The depth is equal to the width (left) or 3 times the width (right). The error model is independent uncorrelated single qubit dephasing noise for each wire touched by a gate. The dephasing probability is 0.001. The deeper the circuit, the more advantageous it is to perform Block-PEC, with already a factor 4 for 9 qubits only in the deeper circuit. The total squarred error for the exponential fit are 0.0004 (left) and 0.0340 (right) while it is 0.0030 (left) and 0.5452 (right) for the quadratic one. This entails the exponential fit a respectively 7.50 and 16.04 better precision than the quadratic one.}
  \label{fig:rbs-swap-net}
\end{figure}

\section{Block-PEC for concrete applications}
\label{sec:applications}

In this section, we consider multiple applications in the context of computational finance (Section~\ref{sub:option-pricing}) and machine learning (Section~\ref{sub:rbs-pyramids}), and which fit in the framework defined by Block-PEC. 

All numerical simulations were implemented in Python, using \texttt{Cirq}~\cite{cirq}. We use the \texttt{mitiq} library~\cite{mitiq} to represent quasi-probabilistic distributions and execute the Monte Carlo simulations they correspond to. We implemented routines which pre-calculate the coefficients of Block-PEC's quasi-probabilistic distributions, for circuits subjected to dephasing noise, 

For convenience, dephasing noise will be modeled by uncorrelated dephasing noise channels, as described in Appendix~\ref{appendix:dephasing}. In practical implementations on cat-qubit architectures, dephasing noise is correlated. In Appendix~\ref{appendix:detailed-performance-analysis}, we show that for circuits subjected to correlated dephasing noise, Block-PEC yields the same performance gains witnessed in Section~\ref{sub:performance-analysis}.

\subsection{Mitigating errors in quantum option pricing} 
\label{sub:option-pricing}
Option pricing is the task of estimating the payoff (value at expiration date) of option contracts. Building on quantum algorithms for amplitude estimation~\cite{low-depth-qae}, Stamatopoulos et al.~defined a family of quantum Monte Carlo algorithms for option pricing~\cite{Stamatopoulos_2020}. Such algorithms aim to calculate the expectation value of the payoff, by applying amplitude estimation to a unitary operator $\mathcal A$ such that
\begin{align}
    \mathcal A \ket{0}^{\otimes (n+1)}
    = \sqrt{1-a}\ket{\psi_0}\ket 0
    + \sqrt{a}\ket{\psi_1}\ket 1
\end{align}
for some normalized states $\ket{\psi_0}$ and $\ket{\psi_1}$, where $a \in [0,1]$ is the expectation value to estimate.
The operator $\mathcal A$ is formed by a data loading unitary $\mathcal P \ket{0}^{\otimes  n} = \sum_x \sqrt{p_x} \ket x$ and a payoff unitary $\mathcal R \ket{x} \ket{0} = \ket{x}\left(\sqrt{f(x)}\ket 1 + \sqrt{1-f(x)} \ket 0\right)$. 

The unitary $\mathcal R$ implements a piece-wise linear function $f: \{0, \ldots, 2^n-1\} \to [0,1]_\RR$ as
\begin{align}
    \mathcal R: \ket{j}\ket{0} \mapsto \ket{j}
    \left(\cos(f_j)\ket{0}+\sin(f_j)\ket{1}\right),
\end{align}
where $f_j = jf_1+f_0$.

The term $f_j$ is implemented by the $j$-th qubit, which controls a $Y$-rotation of angle $\theta_j = 2^jf_1$ on the ancilla qubit, at the exception $f_0$ which is implemented by a $Y$-rotation on the ancilla qubit. $Y$-rotations can be implemented as 
\begin{align}
    Y(\theta) = S H Z(\theta) H S^\dagger
\end{align}
with the phase gate implemented as $S=Z(\pi/2)$, so that Y-rotations are implemented using only Pauli-$Z$ compatible gates. Then, controlled $Y$-rotation can be implemented with only Pauli-$Z$ compatible gates as
\begin{align}
    CY(\theta) = Y\left(\frac{\theta}{2}\right)_k CNOT_{j,k} Y\left(-\frac{\theta}{2}\right)_k CNOT_{j,k}. 
    \label{eq:cy}
\end{align}
We refer the interested reader to~\cite{Stamatopoulos_2020} for a detailed explanation of the definition of the angles $\theta_j$ in order to compute the expectation value of an option with payoff $f$. 

The unitary $\mathcal R$ is implemented by the circuit 
\begin{align}
    \mathcal R = CY(\theta_n)_{n-1,n} \cdots CY(\theta_1)_{0,n} Y(\theta_0)_n. 
\end{align}

Each implementation of $CY(\theta_j)$ given by Equation~\ref{eq:cy} contains one circuit pattern $U_{(a)}$ as $Y\left(\pm\frac{\theta_j}{2}\right)_n CNOT_{j-1,n}$. With $n-1$ repetitions of the circuit pattern $U_{(a)}$ in the implementation of $\mathcal R$, one can expect small gains when using Block-PEC as the sampling cost reduction for such sub-circuits is of lesser magnitude than for other types of sub-circuits.

To quantify the efficiency of Block-PEC on this algorithm, we simulated the mitigation of errors on the payoff unitary $\mathcal R$. The results of our numerical simulations are aggregated in Figure~\ref{fig:violins-opt}. We have simulated the mitigation of errors on $4$-qubit payoff functions, where we chose to implement them using  the unitary $\mathcal R_4 = Y(\theta_0)CY(\theta_1)_{0,4}\cdots CY(\theta_4)_{3,4}$ for randomly chosen angles $\theta_0,\cdots,\theta_4$. Taking $p=10^{-2}$, we computed the average error of PEC and Block-PEC over $1000$ samples. We observed that the sampling variance ratio is $\frac{\sigma^2_{\text{std}}}{\sigma^2_{\text{blk}}} \simeq 1.0375$, which is in line with our expectation of a small sampling variance reduction when using Block-PEC over PEC.

\begin{figure}
    \centering
    \begin{tabular}{l}
        \includegraphics[width=1.1\linewidth]{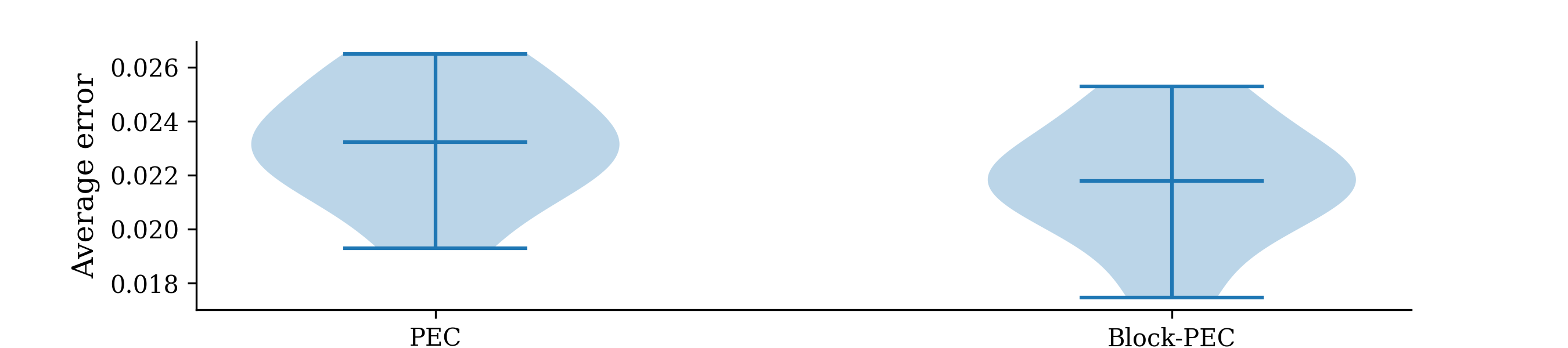}
    \end{tabular}
    \caption{Average sampling error of PEC and Block-PEC on $4$-qubit payoff functions, at $p=0.01$; number of samples = 1000; sample variances: $\sigma^2_{\text{std}} \simeq 5.05e^{-6}$, $\sigma^2_{blk} \simeq 4.87e^{-6}$.}
    \label{fig:violins-opt}
\end{figure}

\subsection{RBS-based parameterized quantum circuits}
\label{sub:rbs-pyramids}
In the use of $RBS$ gates, a common construction involves executing $RBS$ gates for various angles following a pyramidal structure, as shown in Figure~\ref{tab:matchgate-architectures}. For example, any orthogonal neural network can be mapped onto a pyramid of $RBS$ gates acting on unary states~\cite{onn}. The training time of quantum (orthogonal) neural networks is in $O(1/\varepsilon^2)$, where $\varepsilon$ is the accuracy of the estimation of the circuit's output. Mitigating errors increases accuracy, and therefore reduces training time.

It is worth noting that using a unary encoding allows to detect errors by discarding non-unary output states, which can only appear when a bit flip occurred during the computation. While efficient, this is a potentially costly strategy. A low bit-flip error rate architecture makes the cost of this non-unary state error detection strategy negligible. The sampling cost improvement that Block-PEC yields for RBS-based parameterized quantum circuits depends on the particular architecture considered and the implementation of $RBS$ gates. And since the calculation of the coefficients of the quasi-probabilistic distributions for Block-PEC does not depend on the angles of $RBS$ gates, it is sufficient to calculate them once for every $RBS$-based parameterized circuit architecture, further reducing the classical overhead of our Block-PEC approach.

We observed that PEC requires more than twice as many quantum samples as Block-PEC when applied to a 8-qubit $RBS$ pyramid subjected to a phase-flip error probability of $p=0.01$. Since each implementation of $RBS(\theta)$ given by Equation~\ref{eq:cy} contains one circuit pattern $U_{(b)}$ as $XCZ(\theta)_{j,k} CNOT_{j,k}$, one can expect noticeable gains when using Block-PEC. To quantify the efficiency of Block-PEC on this algorithm, we simulated the mitigation of errors on $4$-qubit $RBS$ pyramids at $p=10^{-2}$. We computed the average error of PEC and Block-PEC over 1000 samples. We observed that the sampling variance ratio is $\frac{\sigma^2_{\text{std}}}{\sigma^2_{blk}} \simeq 1.6847$, which suggests a noticeable improvement of the sampling variance when using Block-PEC (see Figure~\ref{fig:violins-rbs}).

\begin{figure}
    \centering
    \begin{tabular}{l}
        \includegraphics[width=1.1\linewidth]{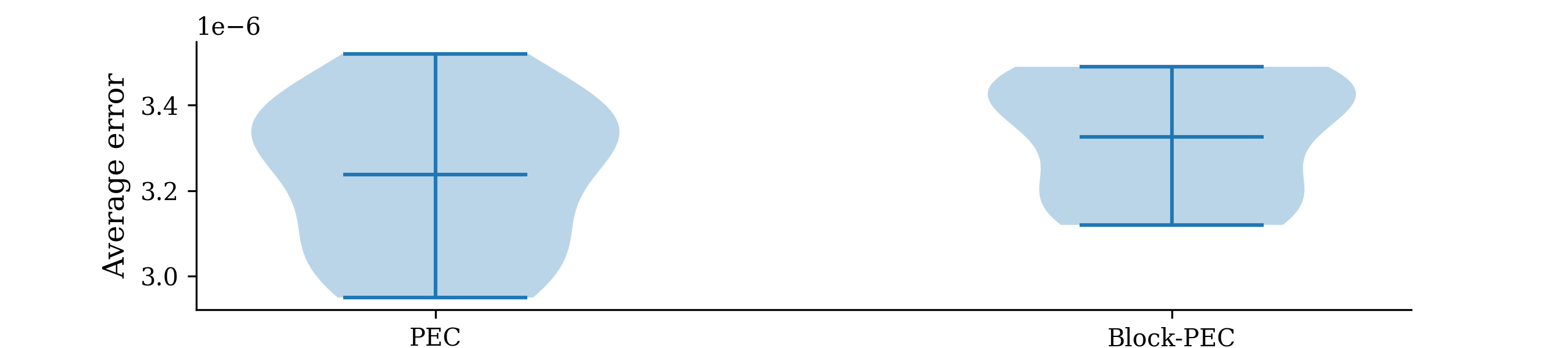}
    \end{tabular}
    \caption{Average sampling error of PEC and Block-PEC on $4$-qubit $RBS$ pyramids at $p=10^{-2}$; number of samples = 1000; sample variances: $\sigma^2_{\text{std}} \simeq 3.51e^{-14}$, $\sigma^2_{blk} \simeq 2.08e^{-14}$.}
    \label{fig:violins-rbs}
\end{figure}

\section{Conclusion}
\label{sec:conclusion}
We have introduced Block-PEC, a variation of the probabilistic error cancellation for error mitigation that takes into account the specifics of the noise model and the propagation of errors to lower the overhead of using noisy quantum computers in lieu of perfect ones. The  model that we consider corresponds to biased noise qubits as implemented by superconducting cat-qubit architectures.

Our technique provides a provable advantage for single layer circuits  and extends to an exponential advantage when several bias-preserving gates follow one another. It constructs an aggregate noise model that can be inverted all at once instead of a per gate basis. This results in a lower sampling cost for achieving a given precision at the cost of classical pre-computation.

We performed numerical simulations for assessing typical performance on several classes of circuits. In accordance with our theoretical results, Block-PEC is most efficient for deeper circuits, as the cost of the classical pre-computation increases exponentially with the width of the circuit but only linearly with its depth.  For example, for 9-qubit circuits with $RBS$ gates in three repetitions of a $SWAP$ network, the gain is already about 4 even at $p=0.001$, and for pyramids of $RBS$ gates it is expected to reach 12 for 30 qubits under the same noise strength. Using these simulation results, we argued that for applications such as option pricing or orthogonal neural networks, Block-PEC leads to a reduction of the sampling variance.

We have shown that this advantage carries over whenever sub-circuits made out of several bias-preserving gates can be identified within the full computation. This is possible by performing a hybrid simulation of the full circuit. The bias-preserving sub-circuits are sampled according to the aggregate noise model defined by Block-PEC while the other gates are simply sampled according to the regular PEC technique. Then, one run of the hybrid simulation corresponds to independently sampling each of the sub-circuits and non bias-preserving gates. Similar results could be obtained for Clifford circuits under Pauli noise.

Moreover, applying our per-block reasoning to Noise-Extended PEC (NEPEC)~\cite{nepec}, which generalizes both PEC and Zero-Noise Extrapolation (ZNE), means that mitigating errors by block can be an efficient strategy in the context of ZNE. 

Finally, this work also exemplifies that cat-qubits architectures can be advantageous to use in the NISQ-regime thanks to their noise model that can be more efficiently handled by error mitigation techniques. As such this raises two related questions: what other error mitigation techniques can benefit from biased noise qubits  and, more generally, are there architectures for which error mitigation techniques can be tailored and made particularly efficient for end-to-end applications? 

\paragraph*{Conflict of Interest.}
The authors declare absence of conflict of interest.

\paragraph*{Acknowledgements.}
This research was supported by Région Île-de-France under project QML-CAT agreement 22002432 PAQ 2022.  The authors would like to thank J\'er\'emie Guillaud and Romain Kukla for numerous fruitful discussions during project meetings.  The authors are grateful to the CLEPS infrastructure from INRIA Paris for providing resources and support in obtaining the numerical results reported here.

\bibliographystyle{acm}
\bibliography{catpec.bib,qubib.bib}

\begin{thebibliography}{10}

\bibitem{Abrams_2020}
{\scshape Abrams, D.~M., Didier, N., Johnson, B.~R., Silva, M. P.~d., and Ryan,
  C.~A.}
\newblock Implementation of xy entangling gates with a single calibrated pulse.
\newblock {\em Nature Electronics 3}, 12 (Nov. 2020), 744–750.

\bibitem{Anselmetti_2021}
{\scshape Anselmetti, G.-L.~R., Wierichs, D., Gogolin, C., and Parrish, R.~M.}
\newblock Local, expressive, quantum-number-preserving vqe ansätze for
  fermionic systems.
\newblock {\em New Journal of Physics 23}, 11 (Nov. 2021), 113010.

\bibitem{arute2020}
{\scshape Arute, F., Arya, K., Babbush, R., Bacon, D., Bardin, J.~C., Barends,
  R., Bengtsson, A., Boixo, S., Broughton, M., Buckley, B.~B., Buell, D.~A.,
  Burkett, B., Bushnell, N., Chen, Y., Chen, Z., Chen, Y.-A., Chiaro, B.,
  Collins, R., Cotton, S.~J., Courtney, W., Demura, S., Derk, A., Dunsworth,
  A., Eppens, D., Eckl, T., Erickson, C., Farhi, E., Fowler, A., Foxen, B.,
  Gidney, C., Giustina, M., Graff, R., Gross, J.~A., Habegger, S., Harrigan,
  M.~P., Ho, A., Hong, S., Huang, T., Huggins, W., Ioffe, L.~B., Isakov, S.~V.,
  Jeffrey, E., Jiang, Z., Jones, C., Kafri, D., Kechedzhi, K., Kelly, J., Kim,
  S., Klimov, P.~V., Korotkov, A.~N., Kostritsa, F., Landhuis, D., Laptev, P.,
  Lindmark, M., Lucero, E., Marthaler, M., Martin, O., Martinis, J.~M.,
  Marusczyk, A., McArdle, S., McClean, J.~R., McCourt, T., McEwen, M., Megrant,
  A., Mejuto-Zaera, C., Mi, X., Mohseni, M., Mruczkiewicz, W., Mutus, J.,
  Naaman, O., Neeley, M., Neill, C., Neven, H., Newman, M., Niu, M.~Y.,
  O'Brien, T.~E., Ostby, E., Pató, B., Petukhov, A., Putterman, H., Quintana,
  C., Reiner, J.-M., Roushan, P., Rubin, N.~C., Sank, D., Satzinger, K.~J.,
  Smelyanskiy, V., Strain, D., Sung, K.~J., Schmitteckert, P., Szalay, M.,
  Tubman, N.~M., Vainsencher, A., White, T., Vogt, N., Yao, Z.~J., Yeh, P.,
  Zalcman, A., and Zanker, S.}
\newblock Observation of separated dynamics of charge and spin in the
  fermi-hubbard model, 2020.

\bibitem{Barkoutsos_2018}
{\scshape Barkoutsos, P.~K., Gonthier, J.~F., Sokolov, I., Moll, N., Salis, G.,
  Fuhrer, A., Ganzhorn, M., Egger, D.~J., Troyer, M., Mezzacapo, A., Filipp,
  S., and Tavernelli, I.}
\newblock {Quantum algorithms for electronic structure calculations:
  Particle-hole Hamiltonian and optimized wave-function expansions}.
\newblock {\em Physical Review A 98}, 2 (Aug. 2018).

\bibitem{100secs-bitflip}
{\scshape Berdou, C., Murani, A., Reglade, U., Smith, W., Villiers, M., Palomo,
  J., Rosticher, M., Denis, A., Morfin, P., Delbecq, M., Kontos, T.,
  Pankratova, N., Rautschke, F., Peronnin, T., Sellem, L.-A., Rouchon, P.,
  Sarlette, A., Mirrahimi, M., Campagne-Ibarcq, P., Jezouin, S., Lescanne, R.,
  and Leghtas, Z.}
\newblock One hundred second bit-flip time in a two-photon dissipative
  oscillator.
\newblock {\em {PRX} Quantum 4}, 2 (jun 2023).

\bibitem{CBBE23quantum}
{\scshape Cai, Z., Babbush, R., Benjamin, S.~C., Endo, S., Huggins, W.~J., Li,
  Y., McClean, J.~R., and O’Brien, T.~E.}
\newblock Quantum error mitigation.
\newblock {\em Reviews of Modern Physics 95}, 4 (Dec. 2023).

\bibitem{qvt}
{\scshape Cherrat, E.~A., Kerenidis, I., Mathur, N., Landman, J., Strahm, M.,
  and Li, Y.~Y.}
\newblock Quantum vision transformers, 2022.

\bibitem{qdeep-hedging}
{\scshape Cherrat, E.~A., Raj, S., Kerenidis, I., Shekhar, A., Wood, B., Dee,
  J., Chakrabarti, S., Chen, R., Herman, D., Hu, S., Minssen, P., Shaydulin,
  R., Sun, Y., Yalovetzky, R., and Pistoia, M.}
\newblock Quantum deep hedging, 2023.

\bibitem{cirq}
{\scshape Developers, C.}
\newblock Cirq, May 2024.

\bibitem{practical-pec}
{\scshape Endo, S., Benjamin, S.~C., and Li, Y.}
\newblock Practical quantum error mitigation for near-future applications.
\newblock {\em Physical Review X 8}, 3 (jul 2018).

\bibitem{EBL18practical}
{\scshape Endo, S., Benjamin, S.~C., and Li, Y.}
\newblock Practical quantum error mitigation for near-future applications.
\newblock {\em Phys. Rev. X 8\/} (Jul 2018), 031027.

\bibitem{ECBY21hybrid}
{\scshape Endo, S., Cai, Z., Benjamin, S.~C., and Yuan, X.}
\newblock Hybrid quantum-classical algorithms and quantum error mitigation.
\newblock {\em Journal of the Physical Society of Japan 90}, 3 (Mar. 2021),
  032001.

\bibitem{vpe-vff}
{\scshape Filip, M.-A., Ramo, D.~M., and Fitzpatrick, N.}
\newblock Variational {P}hase {E}stimation with {V}ariational {F}ast
  {F}orwarding.
\newblock {\em {Quantum} 8\/} (Mar. 2024), 1278.

\bibitem{Gard_2020}
{\scshape Gard, B.~T., Zhu, L., Barron, G.~S., Mayhall, N.~J., Economou, S.~E.,
  and Barnes, E.}
\newblock Efficient symmetry-preserving state preparation circuits for the
  variational quantum eigensolver algorithm.
\newblock {\em npj Quantum Information 6}, 1 (Jan. 2020).

\bibitem{low-depth-qae}
{\scshape Giurgica-Tiron, T., Johri, S., Kerenidis, I., Nguyen, J., Pisenti,
  N., Prakash, A., Sosnova, K., Wright, K., and Zeng, W.}
\newblock Low depth amplitude estimation on a trapped ion quantum computer,
  2021.

\bibitem{QCwithCats}
{\scshape Guillaud, J., Cohen, J., and Mirrahimi, M.}
\newblock Quantum computation with cat qubits.
\newblock {\em arXiv preprint arXiv:2203.03222\/} (2022).

\bibitem{rbs-pde}
{\scshape Jain, N., Landman, J., Mathur, N., and Kerenidis, I.}
\newblock {Quantum Fourier Networks for Solving Parametric PDEs}, 2023.

\bibitem{ncc}
{\scshape Johri, S., Debnath, S., Mocherla, A., Singh, A., Prakash, A., Kim,
  J., and Kerenidis, I.}
\newblock Nearest centroid classification on a trapped ion quantum computer,
  2020.

\bibitem{onn}
{\scshape Kerenidis, I., Landman, J., and Mathur, N.}
\newblock Classical and quantum algorithms for orthogonal neural networks,
  2022.

\bibitem{Kivlichan_2018}
{\scshape Kivlichan, I.~D., McClean, J., Wiebe, N., Gidney, C., Aspuru-Guzik,
  A., Chan, G. K.-L., and Babbush, R.}
\newblock Quantum simulation of electronic structure with linear depth and
  connectivity.
\newblock {\em Physical Review Letters 120}, 11 (Mar. 2018).

\bibitem{mitiq}
{\scshape LaRose, R., Mari, A., Kaiser, S., Karalekas, P.~J., Alves, A.~A.,
  Czarnik, P., El~Mandouh, M., Gordon, M.~H., Hindy, Y., Robertson, A., et~al.}
\newblock Mitiq: A software package for error mitigation on noisy quantum
  computers.
\newblock {\em Quantum 6\/} (2022), 774.

\bibitem{nepec}
{\scshape Mari, A., Shammah, N., and Zeng, W.~J.}
\newblock Extending quantum probabilistic error cancellation by noise scaling.
\newblock {\em Physical Review A 104}, 5 (nov 2021).

\bibitem{Montanaro_2015}
{\scshape Montanaro, A.}
\newblock {Quantum speedup of Monte Carlo methods}.
\newblock {\em Proceedings of the Royal Society A: Mathematical, Physical and
  Engineering Sciences 471}, 2181 (Sept. 2015), 20150301.

\bibitem{OHRW19generalized}
{\scshape O'Gorman, B., Huggins, W.~J., Rieffel, E.~G., and Whaley, K.~B.}
\newblock Generalized swap networks for near-term quantum computing.
\newblock {\em CoRR\/} (2019).

\bibitem{unary-option-pricing}
{\scshape Ramos-Calderer, S., Pérez-Salinas, A., García-Martín, D.,
  Bravo-Prieto, C., Cortada, J., Planagumà, J., and Latorre, J.~I.}
\newblock Quantum unary approach to option pricing.
\newblock {\em Physical Review A 103}, 3 (Mar. 2021).

\bibitem{reglade2023}
{\scshape Réglade, U., Bocquet, A., Gautier, R., Marquet, A., Albertinale, E.,
  Pankratova, N., Hallén, M., Rautschke, F., Sellem, L.-A., Rouchon, P.,
  Sarlette, A., Mirrahimi, M., Campagne-Ibarcq, P., Lescanne, R., Jezouin, S.,
  and Leghtas, Z.}
\newblock Quantum control of a cat-qubit with bit-flip times exceeding ten
  seconds, 2023.

\bibitem{SHH24reducing}
{\scshape Scheiber, T., Haubenwallner, P., and Heller, M.}
\newblock Reducing pec overhead by pauli error propagation.
\newblock {\em CoRR\/} (2024).

\bibitem{Stamatopoulos_2020}
{\scshape Stamatopoulos, N., Egger, D.~J., Sun, Y., Zoufal, C., Iten, R., Shen,
  N., and Woerner, S.}
\newblock Option pricing using quantum computers.
\newblock {\em Quantum 4\/} (July 2020), 291.

\bibitem{SQCB21learning}
{\scshape Strikis, A., Qin, D., Chen, Y., Benjamin, S.~C., and Li, Y.}
\newblock Learning-based quantum error mitigation.
\newblock {\em PRX Quantum 2}, 4 (Nov. 2021).

\bibitem{optimal-pec}
{\scshape Takagi, R.}
\newblock Optimal resource cost for error mitigation.
\newblock {\em Physical Review Research 3}, 3 (aug 2021).

\bibitem{limits-pec}
{\scshape Takagi, R., Endo, S., Minagawa, S., and Gu, M.}
\newblock Fundamental limits of quantum error mitigation.
\newblock {\em npj Quantum Information 8}, 1 (sep 2022).

\bibitem{zne-pec}
{\scshape Temme, K., Bravyi, S., and Gambetta, J.~M.}
\newblock Error mitigation for short-depth quantum circuits.
\newblock {\em Physical Review Letters 119}, 18 (nov 2017).

\bibitem{local-qem}
{\scshape Tran, M.~C., Sharma, K., and Temme, K.}
\newblock Locality and error mitigation of quantum circuits, 2023.

\bibitem{BMKT23probabilistic}
{\scshape van~den Berg, E., Minev, Z.~K., Kandala, A., and Temme, K.}
\newblock Probabilistic error cancellation with sparse pauli–lindblad models
  on noisy quantum processors.
\newblock {\em Nature Physics 19}, 8 (May 2023), 1116–1121.

\end{thebibliography}

\appendix

\section{A primer in probabilistic error cancellation}
\label{app:pec}

Quantum error mitigation techniques aim to minimize the estimation bias when sampling from noisy circuits. The Probabilistic Error Cancellation (PEC~\cite{zne-pec}) approach performs a Monte Carlo simulation of the inverse of the noise model, in order to remove this estimation bias: from a generating set of noisy (implementable) gates, the inverted noise model is approximated as a linear combination of noisy gates. Therefore, bounds on the sampling cost of PEC are tied to the quality of available physical implementations.

The Monte Carlo sampling is applied for a given circuit $\pop U = \pop U_d \circ \cdots \circ \pop U_1$ formed by $d$ gates where each ideal gate $\pop U_l$ is decomposed as the linear combination of implementable noisy gates
\begin{align}
\label{eq:pec-qpd}
\pop U_l = \gamma_l \sum_{1 \leq j \leq J} p_{j,l} s_{j,l} 
\widetilde{\pop B_{j,l}},
\end{align}
and under the assumption that $p_j=|\alpha_{j,l}|/\gamma_l$, $\gamma_l=\sum_j |\alpha_{j,l}|$ and $s_{j,l} \in \{-1,1\}$ are efficiently computable.

This stochastic simulation operates as follows:
\begin{enumerate}
    \item \label{it:mc1} Sample a noisy gate $\widetilde{\pop B_{j,l}}$ for each ideal gate $\pop U_l$ with probability $p_{j,l}$;
    \item \label{it:mc2} Execute the circuit corresponding to the channels $\widetilde{\pop B_{j,d}} \cdots \widetilde{\pop B_{j,1}}$, and measure the observable $O$, obtaining outcome $o$;
    \item\label{it:mc3} Multiply the outcome $o$ by $\gamma_{\text{total}} = \prod_l \gamma_l$ and $\text{sgn}_{\text{total}} = \prod_j s_j$;
    \item \label{it:mc4} Repeat those three steps $S$ times and compute the average $\mu$ of the product $o$.
\end{enumerate}

In this stochastic simulation, an unbiased estimator of $\pop U_l$ is given by
\begin{align}
\widehat{\pop U}_l = \gamma_l s_{j_l} \widetilde{\pop B_{j_l}},
\end{align}
with variance $O(\gamma_l^2)$, so that an unbiased estimator of $\pop U$ is given by 
\begin{align}
\widehat{\pop U}_d \cdots \widehat{\pop U}_1
= \gamma_{\text{total}}\text{sgn}_{\text{total}}
\widetilde{\pop B_{j_d}} \ldots \widetilde{\pop B_{j_1}},
\end{align}
with variance $O(\gamma_{\text{total}}^2)$.

In other words, the stochastic simulation aims to approximate the expectation value $\langle O\rangle_{\text{ideal}} = \Tr[O\pop U(\rho_0)]$ for some observable $O$, given the initial state $\rho_0 = \ketbra{0}^{\otimes n}$. From the quasi-probabilistic distribution defined in Equation~\ref{eq:pec-qpd}, the expectation value of $\pop U$ can be calculated from the expectation value 
\begin{align}
    \langle O\rangle_{j_1,\ldots,j_d} = \Tr[O\widetilde{\pop B_{j_d}} \ldots \widetilde{\pop B_{j_1}}(\rho_0)]
\end{align}
of the noisy gates $\widetilde{\pop B_{j_d}}, \ldots, \widetilde{\pop B_{j_1}}$, that is
\begin{align}
    \langle O\rangle_{\text{ideal}} = \sum_{j_1,\ldots,j_d} \alpha_{j_1}\cdots\alpha_{j_d} \langle O\rangle_{j_1,\ldots,j_d}.
\end{align}

It is the sampling cost $\gamma_{\text{total}}$ which determines the efficiency of PEC as a quantum error mitigation approach. 

Hoeffding's inequality states that the probability that the estimate $\mu$ differs from the true mean value by more than $\delta > 0$ is at most $2\exp\left(-\frac{2S\delta^2}{\gamma_{\text{total}}^2}\right)$,
so that the sample average deviates from the true mean value with variance $\sigma^2 \in O\left(\frac{\gamma_{\text{total}}^2}{S\delta^2}\right)$. Thus the number of samples $S$ required to attain precision $\delta > 0$, with probability at least $1-\varepsilon$, is 
\begin{align}
    S \geq \frac{\gamma_{\text{total}}^2}{2\delta^2} \ln(\frac{2}{\varepsilon}).
\end{align}

\subsection{Solving the noise inversion equation}
Consider the noisy layered circuit
\begin{align}
\widetilde{\pop U} = \widetilde{\pop U}_d \circ \cdots \circ \widetilde{\pop U}_1
\end{align}
defined by a succession of implementable noisy gates, which each satisfy the equality $\widetilde{\pop U}_l= \pop N_l \circ \pop U_l$,
where $\pop N_l$ is the noise map associated to the gate $\pop U_l$.

Now, assuming that $\widetilde{\pop U}_l$ and $\pop U_l$ are known, one can express the inverse $\pop N^{-1}_l = \pop U_l \circ \widetilde{\pop U}_l^{-1}$ of the noise map, 
as a linear combination of unitary channels and preparation channels
\begin{align}
\label{eq:inv-noise}
\pop N_l^{-1} = \sum_j \alpha_{j,l} \pop V_{j,l},
\end{align}
so that each unitary $\pop U_l$ can be rewritten as
\begin{align}
\pop U_l = \pop N_l \circ \pop N^{-1}_l \circ \pop U_l
= \sum_j \alpha_{j,l} \nop B_{j,l},
\end{align}
We define $\widetilde{\pop B_{j,l}} = \pop N_l \circ \pop V_{j,l} \circ \pop U_l$ to be the noisy gate associated to each unitary/preparation channel $\pop V_{j,l}$.

The inverse $\pop N_l^{-1}$ of the noise channel can easily be obtained by observing that determining coefficients $\alpha_{j,l}$ such that Equation~\ref{eq:inv-noise} holds means solving the equation
\begin{align}
\pop I^{\otimes n_l} = \pop N_l \circ \sum_j \alpha_{j,l}  \pop{V}_{j,l},
\end{align}
where $n_l$ is the number of qubits on which the operator $\pop U_l$ acts non-trivially.

\subsection{Parameters influencing the sampling cost}
Minimizing the number of samples needed means solving the following Linear Program (LP) for each $\mathcal U_l$ (which acts non-trivially on $n_l$ qubits):
\begin{align}
&\text{minimize } \sum_{j=1}^J |\alpha_j| \nonumber\\
&\text{subject to }
\pop U_l = \sum_{1 \leq j \leq J} \alpha_j \widetilde{\pop B_{j,l}},
\end{align}
where all the noisy implementable unitary operations $\widetilde{\pop B_{j,l}}$ are in the support of $\pop U_l$ (as a linear combination), assuming that $\pop U_l$ has a non-trivial support and on few qubits.
We take the solution (if feasible) of this LP to be $\gamma_{n_l}$, so that $p_j = \frac{|\alpha_j|}{\gamma_{n_l}}$ and $s_j = \text{sgn}(\alpha_j)$.

We do so for all $\pop U_l$, so that $\pop U$ is obtained as a product of the solutions of the LPs:
\begin{align}
    \pop U = \gamma_{\text{total}} \sum_{j_1,\ldots,j_d} \alpha_{j_1,\ldots,j_d} \widetilde{\pop{B}_{j_d,d}} \cdots \widetilde{\pop{B}_{j_1,1}},
\end{align}
with $\gamma_{\text{total}} = \gamma_{n_d} \cdots \gamma_{n_1}$ and $\alpha_{j_1,\ldots,j_d} = \alpha_{j_d} \cdots \alpha_{j_1}$, 
where $\alpha_{j_l}= p_{j_l} s_{j_l} \gamma_{n_l}$
for each $j_l$,
so that $p_{j_1,\ldots,j_d} = p_{j_d} \cdots p_{j_1}$ and $s_{j_1,\ldots,j_d} = s_{j_d}\cdots s_{j_1}$.

\section{Dephasing noise channels}
\label{appendix:dephasing}

In the present work, dephasing noise channels are linear combinations of Pauli-$Z$ channels. A dephasing noise channel over a set of qubits $A$ can be mathematically represented as linear combination
$\sum_{B \subseteq A} \text{prob}(B) \pop Z_{[B]}$ of Pauli quantum channels $\pop Z_{[B]}$ (where prob is a probability distribution over the partition of $A$), associated to the Pauli unitary $Z_{[B]}$ defined for $B 
\subseteq A$ by 
\begin{align}
    Z_{[B]} = \left(\bigotimes_{i \in B} Z_i\right)
    \otimes \left(\bigotimes_{i \in A \setminus B} I_i\right),
\end{align}
so that $Z_\emptyset = I^{\otimes n}$, and $m=|A|$ is the number of qubits on which $Z_A$ acts non-trivially.

\subsection{Inverting uncorrelated dephasing noise}
Typically (see e.g.~\cite{limits-pec}), uncorrelated dephasing noise channels are defined over a set of qubits $A$ as
\begin{align}
\label{eq:uncorrolated-dephasing-noise}
\pop F_p^{[A]} &= \bigotimes_{j \in A} \left((1-p)\pop I_{j} + p\pop Z_{j}\right)
\end{align}

Now, observe that $\pop F_p^{[A]}$ can be rewritten as 
\begin{align}
    \pop F_p^{[A]} = \sum_{B \subseteq A:|B|=m} (1-p)^{n-m}p^m \pop Z_{[B]},
\end{align}

The inverted noise channel $\left(\pop F_p^{[\{j\}]}\right)^{-1}$ on one qubit is given by
\begin{align}
    \left(\pop F_p^{[\{j\}]}\right)^{-1} = \gamma_1 \left((1-p)\pop I_{j} - p \pop Z_{j}\right)
\end{align}
for $\gamma_k = \frac{1}{(1-2p)^k}$.

Then, for an arbitrary qubit set $A$ of size $n$, we define the inverted noise channel $\left(\pop F_p^{[A]}\right)^{-1}$ by
\begin{align}
    \left(\pop F_p^{[A]}\right)^{-1} &= \gamma_n \left(\bigotimes_{j \in A}\left((1-p)\pop I_{j} - p \pop Z_{j}\right) \right)\nonumber\\
    &= \gamma_n \sum_{B \subseteq A:|B|=m} (1-p)^{n-m}(-p)^m \pop Z_{[B]}.
\end{align}

One can verify that 
\begin{align}
    \gamma_{\text{total}} &= \gamma_n \sum_{m=0}^n {p^m(1-p)^{n-m}} {n \choose m} \nonumber\\
    &= \gamma_n 1 = \gamma_1^n.
\end{align}

\subsection{Approximating impure dephasing noise}

Let us justify our use of a pure dephasing noise model, instead of a model which allows for some depolarizing noise. For simplification we omit in this section the indices of the qubits that noise channels are defined on. 

Consider the local depolarizing noise model given by 
\begin{align}
\pop D^{(1)}_p
&= (1-p)\pop I
+ \frac{p}{3}(\pop X + \pop Y + \pop Z),
\text{ and }
\pop D^{(n)}_p = \left(\pop D_p^{(1)}\right)^{\otimes n}.
\end{align}
One can define an impure local dephasing noise model, which allows for some depolarizing noise, as the noise channels given for any $q \geq 0$ by
\begin{align}
\pop L^{(1)}_{p,q}
&= (1-p)\pop I + p\left(\frac{q}{q+1}\pop Z+\frac{1}{3}\frac{1}{q+1}(\pop X + \pop Y + \pop Z)\right),\\
\pop L_{p,q}^{(n)} &= \left(\pop L_{p,q}^{(1)}\right)^{\otimes n}, 
\end{align}
so that $\pop L_{p,q}^{(n)}$ is the local depolarizing noise for $q=0$ and the local dephasing noise when $q$ tends to infinity.

The distance between the quantum channels $\pop F^{(1)}_p$ and $\pop L^{(1)}_{p,q}$ is defined as the diamond norm of their difference, which is such that
\begin{align}
    \left\|\pop F^{(1)}_p - \pop L^{(1)}_{p,q}\right\|_\diamond = \frac{2p}{3(q+1)} \left\|\pop Z - \frac{1}{2}(\pop X+\pop Y)\right\|_\diamond.
\end{align}
Therefore, the bias incurred by the impure dephasing noise model, quantified as the distance between the pure and impure dephasing noise channels, is proportional to $\frac{2p}{3(q+1)}$.

Observing that
\begin{align}
\left(\pop F^{(1)}_p\right)^{-1}
&= \gamma_1 \left(\left(1-p\right)\pop I-p\pop Z\right), \\
\left(\pop L^{(1)}_{p,q}\right)^{-1}
&= \gamma_1 \left((1-p)\pop I - p \frac{3q+1}{3(q+1)} \pop Z\right. \nonumber\\
&\left.- p \frac{1}{3(q+1)} \left(\pop X+\pop Y\right)\right),
\end{align}
with $\gamma_1 = \frac{1}{1-2p}$,
it follows that the distance $\left\|\left(\pop F^{(1)}_p\right)^{-1} - \left(\pop L^{(1)}_{p,q}\right)^ {-1}\right\|_\diamond$ between the pure and impure dephasing inverted noise channels is proportional to $\gamma_1\frac{2p}{3(q+1)}$.

For cat-qubit architectures, the effective bit flip rate is proportional to $e^{-c\overline{n}}$ where $c$ is a constant and $\overline{n}$ is the mean number of photons (the ``cat size''), which comes at the cost of a higher phase flip rate, with a rate increase proportional to $\overline n$~\cite{QCwithCats}. 

For a realistic description of the noise model of cat-qubit architectures, one could take the noise channel $\pop L^{(1)}_{p,q}$ with $q = \frac{e^{c\overline{n}}}{\overline n} - 1$, so that $\frac{1}{q+1} = \overline n e^{-c\overline{n}}$. However, for $c=2$, $q$ is of the order of $10^3$ and $10^7$ for $\overline n$ equal to $5$ and $10$ respectively (which are both reasonable values for $\overline n$ under current experiments~\cite{100secs-bitflip,QCwithCats,reglade2023}). Therefore, in the context of PEC-like error mitigation protocols, pure dephasing noise channels constitute a good approximation of the noise associated to bias-preserving circuits implemented on cat-qubits.

\section{Bias-preserving gates}
\label{sec:bias-preserving}

Informally, bias-preserving gates are quantum gates which have the property that they do not turn phase-flip errors into bit-flip errors. In other words, phase-flip errors are mapped to phase-flip errors, preserving phase-flip biases. We show that bias-preserving gates have matrix representations which can be described as generalized permutation matrices.

\begin{definition}
A $n$-qubit gate is bias-preserving if its unitary matrix representation $G$ is such that for every $n$-by-$n$ unitary diagonal matrix $D$, there is a $n$-by-$n$ unitary diagonal matrix $D'$ such that $GDG^\dagger=D'$, i.e. $GD=D'G$.
\end{definition}

In other words, a $n$-qubit gate is bias-preserving if its matrix representation is in the normalizing group $N(\mathcal D_n)$ of the group $\mathcal D_n$ of $n$-by-$n$ unitary diagonal matrices. To characterize $N(\mathcal D_n)$, we introduce the notion of generalized permutation matrix.

\begin{definition}
A generalized permutation matrix is a matrix such that each column/row has exactly one non-zero entry.
\end{definition}

Generalized permutation matrices are known to be the normalizer of diagonal matrices, that is, they form the largest subgroup of $GL_n(\CC)$ such that diagonal matrices are normal (i.e. for every $g 
\in GL_n(\CC)$, if $d$ is a $n$-by-$n$ diagonal matrix then so does $gdg^\dagger$). Considering unitary matrices within this context leads to the following proposition.

\begin{proposition}
A gate is bias-preserving if its matrix representation is in $N(\mathcal D_n) = \mathcal P_n \mathcal D_n$.
\end{proposition}

\begin{proof}
Elements of the group $N(\mathcal D_n)$ are unitary generalized permutation matrices. Observe that for every generalized permutation $n$-by-$n$ matrix $M$, there is a permutation $\sigma \in S_n$ and a complex vector $m = (m_1,\ldots,m_n) \in \CC^n$ such that
\begin{align}
M_{i,j} = \left\{
    \begin{array}{ll}
        m_j & \mbox{if } i =\sigma(j) \\
        0 & \mbox{otherwise.}
    \end{array}
\right.
\end{align}
Now, consider two generalized permutation matrices $M$ and $M'$, respectively characterized by the pairs $(\sigma,m)$ and $(\sigma',m')$. Then the multiplication of $M$ by $M'$ is given componentwise by
\begin{align}
(MM')_{i,j} = \left\{
    \begin{array}{ll}
        m_{\sigma'(j)}m'_j & \mbox{if } i =\sigma\sigma'(j) \\
        0 & \mbox{otherwise.}
    \end{array}
\right.
\end{align}
Writing $M_{\sigma,m}$ for the generalized permutation matrix characterized by the permutation $\sigma$ and the complex vector $m$, we observe that
\begin{align}
M_{\sigma,m} = M_{\sigma,(1,\ldots,1)}M_{\text{id},m}
\end{align}
where $M_{\sigma,(1,\ldots,1)}$ is a permutation and $M_{\text{id},m}$ is a diagonal matrix.
In other words, every generalized permutation matrix is the product of a permutation and a diagonal matrix. This statement is in particular true for unitary generalized permutation matrices, for which the characteristic vector $m$ is such that $m_j \neq 0$ and $m_j\overline{m_j} = 1$ (since the eigenvalues of any unitary matrix have modulus 1).
\end{proof}

To put it simply, bias-preserving gates are normalizers of diagonal matrices in the group of unitary matrices. This normalizer is fully characterized by products of permutation matrices and diagonal unitaries. 

\section{Partially bias-preserving gates}
\label{app:partially-bias-preserving}

This section recalls the definition of partially bias-preserving gates, introduces $RBS$ gates, before giving a detailed proof that $RBS$ gates are partially bias-preserving.

\begin{definition}
Consider $\mathcal S$ to be a subspace of the Hilbert space of $n$ quantum states.
A gate is $\mathcal S$-bias-preserving if its unitary matrix representation $G$ is:
\begin{enumerate}
    \item stable under $\mathcal S$, i.e. $\ket \phi \in \mathcal S \implies G\ket \phi \in \mathcal S$;
    \item such that for every $n$-by-$n$ unitary diagonal matrix $D$, there is a $n$-by-$n$ unitary diagonal matrix $D'$ such that, for every state $\ket\phi \in \mathcal S$, $GD\ket\phi=D'G\ket{\phi}$.
\end{enumerate}
\end{definition}

A natural strategy for error mitigation when sampling from such gates is to exclude measurement outcomes which do not belong to the set $\mathcal S$ (i.e. measurement outcomes which are not of the form $\sum_j \alpha_j \ket{\phi_j}$ for $\ket{\phi_j} \in \mathcal S$), and then apply the same error mitigation strategies that we develop for bias-preserving gates. 

\subsection{Reconfigurable Beam Splitter (RBS) gates}
\label{sub:rbs}

The $RBS$ (Reconfigurable Beam Splitter) gate is a gate which performs a rotation in the Hilbert space spanned by the unit vector $\ket{01}$ and $\ket{10}$, so that
\begin{align}
\label{eq:rbs-01}
RBS(\theta)\ket{01} &= \cos(\theta)\ket{01} - \sin(\theta)\ket{10}\\
\label{eq:rbs-10}
RBS(\theta)\ket{10} &= \cos(\theta)\ket{10} + \sin(\theta)\ket{01}
\end{align}
while acting as the identity on $\ket{00}$ and $\ket{11}$.

This $RBS$ gate is particularly useful when its input subspace is restricted to states of Hamming distance 1, as the $RBS$ gate preserves this property. The gate iRBS shares the same properties, since the gate iRBS is such that
\begin{align}
iRBS(\theta)\ket{01} &= \cos(\theta)\ket{01} - i\sin(\theta)\ket{10},\\
iRBS(\theta)\ket{10} &= \cos(\theta)\ket{10} -i\sin(\theta)\ket{01}.
\end{align}
while acting as the identity on $\ket{00}$ and $\ket{11}$.

Using combinations of $RBS$ gates with appropriately defined angles, one can load bit strings into an unary encoded quantum state. Unary data loaders are quantum circuits used to construct states of the form
\begin{align}
\label{eq:unary-data-state}
    \ket x = \frac{1}{\|x\|} \sum_{1 \leq j \leq d} x_j \ket{e_j}
\end{align}
for a given $d$-dimensional vector $\vec x = (x_1,\ldots,x_d)$, where $\ket{e_j}$ is the unary representation of the number $j$, i.e. $e_j = 0^{j-1}10^{d-j}$. They are examples of amplitude encoding routines, as they encode classical data in the amplitudes of a quantum state.

The state in Equation~\ref{eq:unary-data-state} can for example be constructed as follows. Start with the all-$0$-state $\ket{0\ldots0}$ and flip its first qubit. Then, define a set of angles for $RBS$ gates so that when those gates are applied to a qubit $q_n$ and its neighbor $q_{n+1}$ successively, one obtains the target state. Writing $\theta_0 = \arccos(x_1), \theta_1 = \arccos(x_2 / \sin(\theta_0)), \theta_2 = \arccos(x_3 / (\sin(\theta_0)\sin(\theta_1))), \ldots$,
and applying $RBS(\theta_n)$ to the pair of qubits $(q_n,q_{n+1})$, one can construct the target state in Equation~\ref{eq:unary-data-state}. This is the diagonal unary data loader, as represented by the circuit in Figure~\ref{fig:unary-data-loader} for an $8$-point dataset.

\begin{figure}
    \centering
    \includegraphics[width=0.6\linewidth]{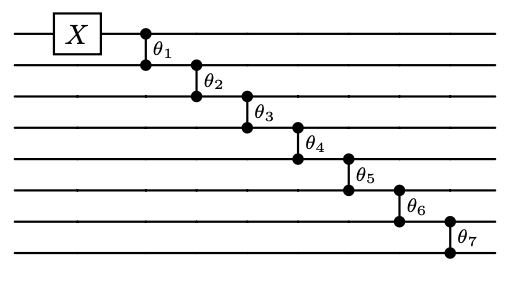}
    \caption{Unary data loaders on 8 qubits (diagonal)}
    \label{fig:unary-data-loader}
\end{figure}

This approach to the encoding of vectors can be extended to matrices. A $n$-by-$m$ matrix $A=[a_{i,j}]_{i,j}$ can be encoded in amplitudes by the quantum state
\begin{align}
    \ket A = \frac{1}{\|A\|} \sum_{i=1}^n\sum_{j=1}^m a_{i,j} \ket{e_j}\ket{e_i}.
\end{align}

The use of $RBS$ gates for unary data loading has found applications to Orthogonal Neural Networks~\cite{onn}, Nearest Centroid Classification~\cite{ncc}, and solving Partial Differential Equations such as Navier-Stokes, using Quantum Fourier Networks~\cite{rbs-pde}.

RBS gates can also be used to implement the inner product of two vectors $x$ and $y$, by loading the vector $y$ and then unloading (i.e.~executing the reversed data loading circuit) the vector $x$. This unary inner product routine is a key element of the implementation of low-depth amplitude estimation~\cite{low-depth-qae}, which finds applications in quantum Monte Carlo simulations~\cite{Montanaro_2015}, leading notably to an unary quantum algorithm for option pricing~\cite{unary-option-pricing}.

\subsection{RBS gates are partially bias-preserving}
While $RBS$ gates have interesting properties with respect to the unary encoding, they are not bias-preserving in the strict sense of the term. Fortunately, $RBS$ gates are partially bias-preserving.

\begin{theorem}
RBS gates $RBS(\theta)$ are $\mathcal S_1$-bias-preserving gates, where $\mathcal S_1$ is the subspace spanned by all $2$-qubit quantum states of Hamming weight $1$ (that is, unary quantum states spanned by $\ket{01}$ and $\ket{10}$). 
\end{theorem}

\begin{proof}
First, observe that any state in $\mathcal S_1$ is a linear combination of unary quantum states and that
\begin{align}
RBS\left(\alpha \ket{01} + \beta \ket{10}\right)
= \alpha' \ket{01} + \beta' \ket{10},
\end{align}
for complex numbers $\alpha,\alpha',\beta,\beta'$ satisfying:
\[
\left\{ 
\begin{array}{ll}
    \alpha' &= \alpha \cos(\theta) + \beta \sin(\theta) \\
    \beta' &= \beta \cos(\theta) - \alpha \sin(\theta).
\end{array}
\right.
\]

Consider an arbitrary state $\ket\phi = \alpha \ket{01} + \beta \ket{10}\in \mathcal S_1$. Then $RBS(\theta)\ket\phi$. Moreover, for every state $\ket\phi \in \mathcal S_1$, it holds that 
\begin{align}
    Z_1Z_2\ket\phi &= -\ket\phi,\\
    Z_1\ket\phi &= -Z_2\ket\phi.
\end{align}

From Equation~\ref{eq:rbs-01} and~\ref{eq:rbs-10}, the action of the $RBS$ gate on the subspace $\mathcal S_1$ can be described as
\begin{align}
\label{eq:rbs-xy}
    RBS(\theta)\ket\phi = \cos(\theta)\ket\phi+\sin(\theta)Z_1X_1X_2\ket\phi,
\end{align}
and therefore 
\begin{align}
\label{eq:rbs-z1}
    RBS(\theta)Z_1\ket\phi &= \cos(\theta)Z_1\ket\phi + \sin(\theta) X_1X_2\ket\phi \nonumber\\
    &= Z_1\left(\cos(\theta)\ket\phi + \sin(\theta)Z_1X_1X_2\ket\phi\right) \nonumber \\
    &= Z_1RBS(\theta)\ket\phi. 
\end{align}

Now, consider an arbitrary diagonal unitary $D$ defined by
\begin{align}
    D &= Z_1(\mu_{1}) Z_2(\mu_{2})ZZ_{1,2}(\mu_{1,2})
\end{align}
where 
\begin{align}
    Z_j(\mu_j) &= \cos(\frac{\mu_j}{2})\ket\phi  - i\sin(\frac{\mu_j}{2})Z_j\ket\phi,\\
    ZZ_{1,2}(\mu_j) &= \cos(\frac{\mu_{1,2}}{2})\ket\phi  - i\sin(\frac{\mu_{1,2}}{2})Z_1Z_{2}\ket\phi.
\end{align}

Observing that 
\begin{align}
    RBS(\theta)ZZ_{1,2}(\mu)\ket\phi 
    &= ZZ_{1,2}(\mu)RBS(\theta)\ket\phi,  \\
    \text{ and } RBS(\theta)Z_j(\mu)\ket\phi 
    &= Z_j(\mu)RBS(\theta)\ket\phi,
\end{align}
it follows that $RBS(\theta)D\ket\phi = D\cdot RBS(\theta)\ket\phi$, and therefore the $RBS$ gate is $\mathcal S_1$-bias-preserving.
\end{proof}

This property holds not only for $(i)RBS(\theta)$, but for other parity-preserving gates such as $\text{A}(\theta,\phi)$ (see e.g.~\cite{Barkoutsos_2018,Gard_2020,vpe-vff}), $\text{Givens}(\theta)$ (see e.g.~\cite{Anselmetti_2021,Kivlichan_2018}), $\text{XY}(\theta)$ (see e.g.~\cite{Abrams_2020}). 

\section{Detailed calculations of performance analysis of Block-PEC on bias-preser\-ving gates}
\label{appendix:detailed-performance-analysis}

In this section, we calculate the Block-PEC sampling cost of the circuits  $U_{(a)}=CNOT_{1,2}U_2$, $U_{(b)}=CNOT_{1,2}U_{1,2}$, $U_{(c)}=CNOT_{i,j}U_{j,k}$ ($i \neq k$), directly from their quasi-probabilistic distribution, assuming $p > 0$.

We write $\alpha_0 = \frac{1-p}{1-2p}$ and $\alpha_1 = -\frac{p}{1-2p}$, so that
\begin{align}
    \pop U_2 &= \alpha_0 \pop I_2 \circ \pop U_2 + \alpha_1 \pop Z_2 \circ \pop U_2\\
    \pop U_{1,2} &= \alpha_0^2 \pop I_1 \pop I_2 \circ \pop U_{1,2} \nonumber\\
    &+ \alpha_0\alpha_1 (\pop I_1\pop Z_2 + \pop Z_1\pop I_2) \circ \pop U_{1,2}
    + \alpha_1^2 \pop Z_1\pop Z_2 \circ \pop U_{1,2}
\end{align}

For the circuit $U_{(a)}$, $\gamma_{\text{blk}}(U_{(a)}) = \beta_0 - 3 \beta_1$ with $\beta_0 = \alpha_0^3 + \alpha_1^3 = ((1-p)^3 - p^3) / (1-2p)^3 > 0$ (associated to the channel $\pop I_2$) and $\beta_1 = \alpha_0\alpha_1^2 + \alpha_1\alpha_0^2 = (p^2(1-p) - p(1-p)^2) / (1-2p)^3 < 0$ (associated to $\pop Z_2$), so that 
\begin{align}
\label{eq:blk-ua}
    \gamma_{\text{blk}}(U_{(a)}) = \frac{1+2p-2p^2}{(1-2p)^2} < \frac{1}{(1-2p)^3} = \gamma_{\text{std}}(U_{(a)}).
\end{align}

For the circuit $U_{(b)}$, $\gamma_{\text{blk}}(U_{(b)}) = \beta_0-\beta_1-2\beta_2$, with $\beta_0 = \alpha_0^4 + \alpha_0\alpha_1^3+\alpha_1^2\alpha_0^2+\alpha_1^2\alpha_0\alpha_1 > 0$ (associated to the channel $\pop I_1\pop I_2$), $\beta_1 = \alpha_0^3\alpha_1 + \alpha_1\alpha_0^3 + \alpha_0^2\alpha_1^2 + \alpha_1^4 < 0$ (associated to the channel $\pop Z_1\pop I_2$), $\beta_2 = 2\alpha_0^2\alpha_1^2 + \alpha_0\alpha_ 1^3 + \alpha_1\alpha_0^3 < 0$ (associated to the channels $\pop I_1\pop Z_2$ and $\pop Z_1\pop Z_2$), so that
\begin{align}
\label{eq:blk-ub}
    \gamma_{\text{blk}}(U_{(b)}) &= \frac{1+2p-6p^2+4p^3}{(1-2p)^3} < \frac{1}{(1-2p)^4} = \gamma_{\text{std}}(U_{(b)}).
\end{align}

For the circuit $U_{(c)}$, $\gamma_{\text{blk}}(U_{(c)}) = \beta_0 - \beta_1 - 3(\beta_2 - \beta_3)$, with $\beta_0 = \alpha_0^4+\alpha_1^3\alpha_0 > 0$ (associated to the channel $\pop I_1\pop I_2\pop I_3$), $\beta_1 = \alpha_0^3\alpha_1 + \alpha_1^4 < 0$ (associated to the channel $\pop I_1\pop I_2\pop Z_3$), $\beta_2 = \alpha_0^3\alpha_1 + \alpha_1^2\alpha_0^2 < 0$ (associated to the channels $\pop I_1\pop Z_2\pop I_3$, $\pop Z_1\pop I_2\pop I_3$, $\pop Z_1\pop Z_2\pop I_3$), $\beta_3 = \alpha_1^3\alpha_0 + \alpha_1^2\alpha_0^2 > 0$ (associated to the channels $\pop I_1\pop Z_2\pop Z_3$, $\pop Z_1\pop I_2\pop Z_3$, $\pop Z_1\pop Z_2\pop Z_3$), so that
\begin{align}
\label{eq:blk-uc}
    \gamma_{\text{blk}}(U_{(c)}) &= \frac{1+2p-2p^2}{(1-2p)^3} < \frac{1}{(1-2p)^4} = \gamma_{\text{std}}(U_{(c)}).
\end{align}

\paragraph*{Correlated dephasing noise}
The per-block approach to PEC still yields an improvement in terms of sampling cost when dealing with correlated dephasing noise. Indeed, consider the correlated 
dephasing noise channel over a $n$-qubit set $A$ defined as:
\begin{align}
    \pop C^{[A]}_p = (1-p)\pop I^{\otimes n} + \sum_{I^{\otimes n} \neq B \subseteq A} \frac{1}{2^n-1} p Z_{[B]}, 
\end{align}
inverted as the noise channel
\begin{align}
    \left(\pop C^{[A]}_p\right)^{-1} = \gamma_1\left((1-p) \pop I^{\otimes n} - \sum_{I^{\otimes n} \neq B \subseteq A} \frac{1}{2^n-1} p Z_{[B]}\right), 
\end{align}
with $\gamma_1 = \frac{1}{1-2p}$.

Writing $\alpha_0 = \gamma_1(1-p)$ and $\alpha_1 = -\frac{p}{3}\gamma_1$, we observe for the circuit $U_{(b)}$ that $\gamma_{\text{blk}}(U_{(b)}) = \beta_0 - 3\beta_1$ with $\beta_0 = \alpha_0^2 + 3\alpha_1^2 > 0$ (associated to the channel $\pop I_1\pop I_2$) and $\beta_1 = 2\alpha_0\alpha_1 + 2\alpha_1^2 < 0$ (associated to the channels $\pop Z_1\pop I_2$,  $\pop I_1\pop Z_2$ and  $\pop Z_1\pop Z_2$), so that
\begin{align}
    \gamma_{\text{blk}}(U_{(b)}) = 
    \frac{3-4p^2}{3(1-2p)}
    < \frac{1}{(1-2p)^2} = \gamma_1^2 = \gamma_{\text{std}}(U_{(b)})
\end{align}
for correlated dephasing noise.

A similar analysis can be done for the circuits $U_{(a)}$ and $U_{(c)}$.

\end{document}